\documentclass[a4paper,leqno,twoside]{amsart}
\usepackage{amsmath,amssymb,amsthm,enumerate,hyperref,color}
\usepackage{graphicx,epsfig}

\newtheorem{theorem}{Theorem}[section]

\newtheorem{remark}{Remark}[section]
\numberwithin{equation}{section}

\newif\ifcomment \commentfalse
\def\commentON{\commenttrue}

\long\outer\def\BC#1\EC{\ifcomment \sloppy \par \# \ldots\dotfill
{\em #1} \dotfill \# \par \fi } \commentON

\newcommand{\remove}[1]{}

\newcommand{\ep}{\varepsilon}

\newcommand{\R}{{\mathbb R}}

         \newcommand{\jcal}{{\mathcal J}}

\title[A rare mutation model in heterogeneous environment]{A rare mutation model in a spatial heterogeneous environment: \\  the Hawk and Dove game}
\author[A.L.~Amadori, R.~Natalini, D.~Palmigiani]{Anna Lisa Amadori$^{1,2}$, Roberto Natalini$^{2}$, Davide Palmigiani$^{2,3}$}
\address{$^1$Dipartimento di Scienze Applicate, Universit\`a di Napoli \lq\lq Parthenope''}
\address{$^2$ Istituto per le Applicazioni del Calcolo \lq\lq M. Picone'', Consiglio Nazionale delle Ricerche}
\address{$^3$ Dipartimento di Matematica, Universit\`a di Roma ``La Sapienza''}
\thanks{The first author is  member of the GNAMPA group of the Istituto Nazionale di Alta Matematica (INdAM)}

\begin{document}

\begin{abstract}
 We propose a stochastic model in evolutionary game theory where individuals (or subpopulations) can mutate changing their strategies randomly (but rarely) and explore the external environment. This environment affects the selective pressure by modifying the payoff arising from the interactions between strategies. We derive a Fokker-Plank integro-differential equation and provide Monte-Carlo simulations for the Hawks vs Doves game. In particular we show that, in some cases, taking into account the external environment favors the persistence of the low-fitness strategy.
\end{abstract}


 \maketitle

\section{Introduction} 

Evolutionary Dynamics describes biological systems subject to Darwinian Evolution by taking into account the main mechanisms and phenomena of Evolution itself. In \cite{MSP73}, Maynard Smith and Price propose a first instance of this approach by considering a population modified according to the replicator dynamics.
A population 
is formed by $d$ \emph{types}, or \emph{behaviors},
$E_{1},\ldots,E_{d}$, with fractions corresponding to relative abundance
in the vector $x=\left(x_{1},\ldots,x_{d}\right)$, which corresponds
to a point in the simplex
\[\mathcal{S}^d= \Big\{ x=(x_1,\dots x_d)\in\R^d \, : \, x_k\ge 0 , \, \sum\limits_{k=1}^d x_k =1 \Big\}.\]
The selection and adaptation mechanism is described by means of a system of differential equations in the following form:
\begin{equation}\label{eq:replicator dynamics}
\dfrac{\dot{x}_{k}}{x_{k}}=f_{k}(x)-\bar{f}(x),
\end{equation}
as $k=1,\dots d$.
The rate of increment ${\dot{x}_{k}}/{x_{k}}$ of the type $E_{k}$
is given by its absolute fitness, denoted with $f_{k}$, balanced
with the average fitness of the population $\bar{f}$, which has the
form
\[
\bar{f}(x)=\sum_{k=1}^{d}x_{k}f_{k}(x).
\]
In evolutionary game theory the vector of absolute fitness $f=\left(f_{1},\ldots,f_{d}\right)$ is defined as a linear
function of the population $x$, i.e.
\[ f(x)=\mathcal{U}x,\] where
$\mathcal{U}$ is the matrix of payoff that rules the interplay between different strategists. 
In this regard, the fitness of the type
$E_{k}$ is defined as the result that an individual of
that type gets colliding against another individual on average,  i.e.
\[
f_{k}(x)=\left[\mathcal{U}x\right]_{k}=\sum\limits_{i=1}^{d}u_{ki}x_{i}.
\]
{However, it is clear that the basic element for the generation of evolutionary novelties are mutations.
The first attempt of giving account of mutations, dating back to the 1970s, is the so called “quasispecies equation”, where the growth rate of each species is modified by considering the dispersion due to the birth of mutated offspring. The same underlying idea has been included in the evolutionary games setting in \cite{SS92mutation} with the “replicator-mutator” equation:
\begin{equation}\label{quasispecies}
\dot{x}_k  = \sum\limits_{i=1}^{d}  f_i  \, q_{ik}\, x_i -   \bar{f} x_k .
\end{equation}
Here the coefficient $q_{i k}$ express the proportion of offspring of $k$-type from a progenitor $i$, which shows up at any  procreation. 
An important aspect of mutations stands in their randomness, which is quite underrated in \eqref{quasispecies}. Since then many more refined models have been proposed to put into the right light randomness; we refer for instance to \cite{CFM06} showing that one single stochastic microscopic process can generate  different macroscopic models of adaptive evolution. 
More recently, in \cite{ACNT15}, it has been proposed a macroscopic stochastic model where mutations occur at a different time scale than selection. This approach goes into the direction of adaptive dynamics, but differentiates from trait substitution sequence because it is not assumed that there is complete adaptation (namely invasion or extinction of the mutant trait) between subsequent mutations. This aspect is  well-suited especially within the framework of social dilemma, or more generally when the type $E_i$ are read as strategies more then species or phenotypes.}
See also the numerical paper \cite{ABN16}, focused on Prisoner's Dilemma.

In this paper we take a step further and address our attention to  the environment, seen as a place where individuals can
evolve but also as a factor that can influence the dynamics of interaction
between strategists.
The model presented in \cite{ACNT15} is then expanded to take into account how the natural environment can modify the interactions between individuals, changing selective pressures; we add a new variable $y \in \R^N$ to the variable $x$, in the simplex, so that the status of the population is described by the pair $(x,y)$. The new variable $y$ stands for the position of the population or, more widely, for an external parameter that affects the results of the interplay between strategies.
It changes according to a velocity, partly deterministic, partly stochastic, and influences the selection mechanism because the payoff matrix depends on $y$.
\\
The interpretation is twofold: from one hand, one could start with a single population with status $(x_0,y_0)$ and look at how it evolves exploring an heterogeneous environment, subject to selection and mutation. On the other hand one could also start with a sort of meta-population composed by different subpopulations  and observe how its distribution evolves. In this sense selection only acts inside subpopulations: melting of different subpopulation is not easily introduced into a frequency-based model like the present one. This aspect certainly deserves a further investigation.

In the following section \ref{ModelloSDE} we recall the stochastic model for replicator dynamics with point-type mutations introduced in \cite{ACNT15}. With the aim of performing Monte-Carlo simulations, we give an alternative (but equivalent) description of the process by using a single Poisson random measure.
Starting from this description, we generate an algorithm to simulate our process.
Next, the spatial environment is introduced as a further stochastic variable, whose dynamics is ruled by a SDE. Therefore, we end up with two coupled SDE for the  character-position variables $(x,y)$: see \eqref{Replicator mutator space x}, \eqref{Replicator mutator space}.
  
  In Section \ref{fokkerplank} we derive a Fokker-Plank integro-differential equation for \eqref{Replicator mutator space x}, \eqref{Replicator mutator space}, (see \eqref{eq:FP} later on). The classical regularity assumptions requested by the Hormander theory are not satisfied because of the presence of a non-local term, which is the deterministic counterpart of the point process modeling mutations. We therefore read it in the viscosity sense, even if the problem \eqref{eq:FP} does not fit plainly in the standard framework of viscosity solutions for integro-differential equations: the main difficulty comes from the domain where it is set, which is closed. Actually, the model does not justify any attempt to impose a boundary condition. Moreover the nonlocal term does not depend continuously on $x$.
These difficulties are overcome by extending in a suitable way the problem to the whole space \eqref{FP2} and noticing that the produced solution can actually be interpreted as a probability density for the couple character-position $(x,y)$.

Finally, Section \ref{Simulazioni} provides numerical simulations for the probability density obtained by a Monte-Carlo method starting from the stochastic system \eqref{Replicator mutator space x}, \eqref{Replicator mutator space}, using Matlab.
We consider the two strategist game Hawks vs Doves, used by Maynard Smith to explain the high frequency of conventional displays, rather than  all-out fight, among animals (especially within heavily armed species) \cite{SHbook}.
  We modify the standard model by assuming that the cost for fighting changes according to the location, and show how an equilibrium for the replicator-mutator equation can be disrupted by effect of random motion or mutations.
  We also notice that in some particular cases, space itself allows for the survival of the low fitness species.

\section{A stochastic model for mutations in heterogeneous environment}\label{ModelloSDE}

We propose to describe the frequencies of different phenotypes in the  population according to a stochastic differential equations (SDE) in the general framework
\begin{equation}\label{eq.stocastica generale}
X_{t}=X_{0}+\int_{0}^{t}a\left(X_{s}\right)ds+\int_{0}^{t}b(X_{s})dW(s)+\int_{0}^{t}\int_{E}K(X_{s^{-}},\xi)\mathcal{N}(ds\times d\xi).
\end{equation}
Here $X_t$ is a process on a probability space $\left(\Omega,\mathcal{F},\mathbb{P}\right)$,
where $a, b, K$ are Borel measurable functions of appropriate dimensions.
$W(s)$  is a standard Brownian motion and  $P(s)$ is a Poisson point process with random measure $\mathcal{N}(ds\times d\xi)$ on $\mathbb{R}^{+}\times E$, with mean measure $l\times\nu$, $l$ Lebesgue measure on $\mathbb{R}^{+}$, $\nu$ a $\sigma$-finite measure on a measurable space $\left(E,\mathcal{E}\right)$.

The process of classic replicator dynamics (\ref{eq:replicator dynamics})
is obtained when $X=(x_1,\dots x_d)$ is the vector of relative frequencies of $d$ various phenotypes, $a$ is the vector of relative fitness, i.e. $a(X)=\left(\ldots,a_{k}(X),\ldots\right)$,
with 
\[
a_{k}(X)=x_{k}\left(f_{k}(X)-\bar{f}(X)\right),
\]
and $b$ and $K$ are null, so that 
(\ref{eq.stocastica generale}) is totally deterministic.

In  \cite{ACNT15},  mutations are described by means of a pure point process that alters replicator dynamics and the Brownian motion term is zero ($b=0$).
Any mutation has a fixed progenitor (type $i$) and a unique descendant (type $j$): this gives $2\binom{d}{2}=d(d-1)$ different mutations, precisely all those that transform a type $i$ in a type $j$ as
\[
(i,j)\in I=\left\{ \left(i,j\right)\in\left\{ 1,\dots,d\right\} ^{2};i\neq j\right\} .
\]
The mutation from type $i$ to type $j$ is driven by  a non-homogeneous point process $N_{t}^{ij}$ with stochastic
intensity $\lambda_{ij}f_{i}(X_{t^-})$. The process $N_{t}^{ij}$ makes
unit jumps with a frequency depending on the process itself, according
to the ``genetic distance'' between the types $i$ and $j$ ($\lambda_{ij}$) and
the fitness of $i$ ($f_i$): the higher the fitness, the higher the rate
of reproduction of individuals of that kind, the more they will suffer
mutations.
A further coefficient $\gamma_{ij}\in(0,1)$ measures the proportion of individuals involved in mutations: the population of type $i$ decreases
by a fraction $\gamma_{ij}x_{i}$ , while the population of type $j$
increases by the same amount. This yields a jump of the population frequency vector of size $ \gamma_{ij}x_{i} ({e_j}-{e_i})$, ${e_i}$ standing for the unit vector pointing in the direction $i$. 
The resulting SDE is 
\begin{equation}\label{processo Rare Mut}
x_{k,t}=x_{k}(0)+\int_{0}^{t}a(X_s)ds+\sum\limits_{i\neq k}\int_{0}^{t}\gamma_{ik}x_{i,t}dN_{t}^{ik}-\sum\limits_{i\neq k}\int_{0}^{t}\gamma_{ki}x_{k,t}dN_{t}^{ki} .
\end{equation}
Let us notice by now that the number of variables depicting the character can be reduced by observing that $x_d=1-\sum\limits_{i=1}^{d-1}x_i$ and setting the problem in the closed set 
\[ \Sigma^d = \{ (x_1,\dots x_{d-1}) \, : \, x_i \ge 0 , \, \sum\limits_{i=1}^{d-1}x_i \le 1\} .\]
With a little abuse of notations we shall continue to write $x\in\Sigma^d$ and 
\begin{align*}
  f_k(x)= & f_k (x_1,\dots x_{d-1},1-\sum\limits_{i=1}^{d-1}x_i) , \\
  a_k(x)= &a_{k}(x_1,\dots x_{d-1},1-\sum\limits_{i=1}^{d-1}x_i) =  (f_k-f_d)  (1-x_k) x_k - \sum\limits_{\substack{i=1\\ i\neq k}}^{d-1}  (f_i-f_d) x_ix_k. 
\end{align*}

In the same paper \cite{ACNT15}, a Kolmogorov integro-differential equation describing the expected frequencies is derived and investigated analytically, with particular attention to the long term equilibrium.
Analytical investigation  is satisfactory in the case of constant fitness (quasispecies equation), but there are some gaps concerning variable fitness, that has been tackled by a numeric approach in the subsequent paper \cite{ABN16}. 
In the present work we are  mainly concerned with Monte-Carlo simulations.
That is why, before enriching the model by including the effect of heterogeneous environment, it is worth giving an alternative description and present an algorithmic approach. \\
The SDE \eqref{processo Rare Mut} can be written  in standard form \eqref{eq.stocastica generale} 
by taking
$d(d-1)$ independent Poisson random measures $\mathcal{N}_{ij}(ds\times d\xi)$
on $\mathbb{R}^{+}\times \R^+$,  defining the amplitudes of jumps as
\[
K_{ij}(X,\xi)=\gamma_{ij}x_{i}(e_{j}-e_{i})\mathbb{I}_{[0,\lambda_{ij}f_{i}(X))}(\xi),
\]
and then invoking the Poisson embedding \cite{BIB}.

 It is possible to set up an equivalent model with only one random measure $\mathcal{N}(ds\times d\xi)$ on $\mathbb{R}^{+}\times E$ with $E=\R^+\times[0,1]$.
 To this aim we look at the sum of the stochastic intensity of each individual process
\[
\Lambda(X)=\sum_{i\neq j}\lambda_{ij}f_{i}(X),
\]
split the unit interval into  $d(d-1)$ disjoint intervals $\mathcal{I}_{ij}$ of length $\lambda_{ij}f_{i}(X)/\Lambda(X)$,
and take the amplitude of jumps as 
\begin{equation}
K\left(X,\xi\right)=K\left(X,u,\theta\right)=\mathbb{I}_{\left[0,\Lambda(X)\right)}(\theta)\sum\limits_{i\neq j}\left[\gamma_{ij}x_{i}\left(e_{j}-e_{i}\right)\mathbb{I}_{\mathcal{I}_{ij}}(u)\right].\label{K}
\end{equation}
The two processes just described 
coincide indeed.
Actually the infinitesimal generator of the last one is 
\begin{align*}
\int_{(0,1)}\int_{\R}\left[\phi\left(X+\mathbb{I}_{(0,\Lambda(X)]}(\theta)\sum_{i\neq j}\gamma_{ij}x_{i}(e_{j}-e_{i})\mathbb{I}_{\mathcal{I}_{ij}}(u)\right)-\phi(X)\right]du\,d\theta
\\
=\Lambda(X)\int_{(0,1)}\left[\phi\left(X+\sum_{i\neq j}\gamma_{ij}x_{i}(e_{j}-e_{i})\mathbb{I}_{\mathcal{I}_{ij}}(u)\right)-\phi(X)\right]du
\\
=\Lambda(X)\sum_{i\neq j}\int_{\mathcal{I}_{ij}}\left[\phi\left(X+\gamma_{ij}x_{i}(e_{j}-e_{i})\right)-\phi(X)\right]du
\\
=\Lambda(X)\sum_{i\neq j}\left|\mathcal{I}_{ij}\right|\left[\phi\left(X+\gamma_{ij}x_{i}(e_{j}-e_{i})\right)-\phi(X)\right]
\\
=\sum_{i\neq j}\lambda_{ij}f_{i}(X)\left[\phi\left(X+\gamma_{ij}x_{i}(e_{j}-e_{i})\right)-\phi(X)\right] ,
\end{align*}
i.e.~the same infinitesimal generator of the first one. 

In view of Monte Carlo approximations let us give an intuitive interpretation, albeit rigorous, of the two settings, based on the existence theorem for Poisson random measures in \cite{KIS}. 
Let $T>0$ a fixed time horizon and 
\[\lambda_{ij}^{\mathrm{max}}= \max\limits_{X} \lambda_{ij} f_i(X).\]
The evolution process can be simulated by the following steps
\begin{enumerate}[i)]
\item Compute a priori the $d(d-1)$ 
independent homogeneous Poisson processes
with intensities $\lambda_{ij}^{\mathrm{max}}$ with jump times
$T_{n}^{ij},$ lower than $T$: these processes jump more
often than the target ones;
\item Simulate  the ODE related to replicator dynamics till the
minimum between the times $T_{n}^{ij}$, say it $t = T_{1}^{\hat i \hat j}$;
\item Extract a uniform random number $\xi\in [0,1]$
  and use an acceptance-rejection method:
  \begin{enumerate}[{iii.}a)]
  \item if $\lambda_{\hat i\hat j}^{\mathrm{max}} \xi>\lambda_{\hat i\hat j} f_{\hat i}(X_{t^-})$ no jumps occur;
  \item otherwise shift a quantity $\gamma_{\hat i \hat j}x_{\hat i,t^-}$ from $\hat i$ to $\hat j$.
    \end{enumerate}
  \item Restart from step {ii)}.
\end{enumerate}

Concerning the second setting, let
\[ \Lambda^{\mathrm{max}} = \max\limits_{X} \Lambda (X),\]
and execute the following steps
\begin{enumerate}[i)]
\item Build a priori an homogeneous Poisson process with intensity $\Lambda^{\mathrm{max}}$, whose  jump times  will be denoted by $T_{n}$ lower than $T$;
\item Simulate the replicator dynamics till $T_1$;
\item Extract uniformly a random number $\xi\in [0,1]$;
  \begin{enumerate}[{iii.}a)]
  \item if $\Lambda^{\mathrm{max}}\xi> \Lambda(X_{T_1^-}) $ no jump occur,
  \item if $\Lambda^{\mathrm{max}}\xi\le \Lambda(X_{T_1^-}) $ a jump occurs indeed. \\ To decide which kind of mutation occurs,  extract another random number $u\sim{\mathrm{Unif}}(0,1)$ and look at which interval $\mathcal{I}_{\hat i\hat j}$ it belongs (it is possible because the sets $\mathcal{I}_{ij}$ form a partition of $[0,1]$).
    
Then shift a quantity $\gamma_{\hat i\hat j}x_{\hat i,T_1^-}$ from $\hat i$ to $\hat j$.
  \end{enumerate}
\item Restart  from step {ii)}.
\end{enumerate}

\subsection{Heterogeneous environment}

{In the present model the only observed variables are the frequencies of the various phenotypes, as well as in the classical replicator equation.
The rules of the play are fixed once and for all by means of the payoff matrix ${\mathcal U}$, and nothing depends by the physical position of the population, as if the individuals were not able to move, or if the environment were completely homogeneous. A more realistic picture has to take into account that environmental changes affect the results of interaction between different behaviors.
\\
To introduce heterogeneous environment we increase the observed variables so that the status of the population (or of a sub-population) is described by a pair $X=\left(x,y\right)$: as before $x=(x_{1},\ldots,x_{d-1})\in\Sigma^{d}$ 
stands for the \emph{character} of the population, each $x_i$ being the fraction of individuals of type $E_{i}$ (and $x_d=1-\sum\limits_{i=1}^{d-1}x_i$ the fraction of type $E_d$), while the new variable  $y\in\mathbb{R}^{N}$ stands for the \emph{position} of the population.
More widely this new variable can be seen as an external parameter that affects the results of the interplay between strategies.
  The payoff matrix depends by $y$, i.e. $\mathcal{U}=\mathcal{U}(y)$, consequently also the respective fitness
  \begin{align*}
    f_k(x,y)= & \sum\limits_{i=1}^{d-1} u_{ki}(y)x_i + u_{id}(y)(1-\sum\limits_{i=1}^{d-1} x_i) 
    \end{align*}
    varies with $y$.
\\
The  character $x$ evolves according to a suitable version of equation (\ref{processo Rare Mut}):
\begin{equation} \label{Replicator mutator space x}
 \displaystyle  x_t=x_0+\int_{0}^{t}a\left(x_s,y_s\right)ds+\int_{0}^{t}\int_{E}K(x_{s^-},y_s,\xi)\mathcal{N}(ds\times d\xi) .
\end{equation}
Here
\begin{itemize}
\item
  $a\in R^{d-1}$ stands for the \emph{vector field of the replicator dynamics}. It has the same structure as in the former case, but with
    an important difference:  the fitness are allowed to depend from $y$, so that 
\[ a_k(x,y)=x_k (f(x,y)-\bar f(x,y)) \quad \text{ as } k=1,\dots d.\]
\item The jump amplitude $K$ and the Random measure $\mathcal N$ describe the \emph{mutation process} as before. The location $y$ affects the mutation process through the fitness, as
 \[ \Lambda(x,y)=\sum_{i\neq j}\lambda_{ij}f_{i}(x,y), \]
\[
K\left(x,y,u,\theta\right)=\mathbb{I}_{\left[0,\Lambda(x,y)\right)}(\theta)\sum\limits_{i\neq j}\gamma_{ij}x_{i}\left(e_{j}-e_{i}\right)\mathbb{I}_{\mathcal{I}_{ij}(x,y)}(u),
  \]
  where the intervals $\mathcal{I}_{ij}(x,y)$ have length equal to $\lambda_{ij}f_{i}(x,y)/\Lambda(x,y)$ and form a partition of the unit interval, as $i\neq j \in \{1,\dots, d\}$.
\end{itemize}
The position $y$ changes according to a diffusion with drift:
\begin{equation}\label{Replicator mutator space}
 \displaystyle y_t=y_0+\int_{0}^{t}v\left(x_s,y_s\right)ds+\int_{0}^{t}\sigma\left(x_s,y_s\right) dW_s ,
\end{equation}
where
\begin{itemize}
  \item
    $v\in R^N$ stands for the \emph{velocity field of the population}.
    For any given $y$, $v(e_i,y)$ is the drift of the type $E_i$, while a composite population described by the character $x$ is inclined to move according to $v(x,y)$.
\item $\sigma$ is an $N\times N$  matrix and $W_s$ is an $N$-dimensional Brownian motion, describing the random component of the displacement.
\end{itemize}
The model \eqref{Replicator mutator space x}, \eqref{Replicator mutator space}  can be read in two ways. From one hand, one could start with a single population with initial status $(x_0,y_0)$ and look at how it evolves exploring an heterogeneous environment, subject to selection and mutation. But one could also start with a sort of metapopulation composed by different subpopulations distributed like $x_0(y)$ and observe how its distribution evolves. 
In this sense selection only acts inside subpopulations: melting of different subpopulations is not easily introduced into a frequency-based model like the present one. This aspect certainly deserves a further investigation.}

The well posedness of the  process \eqref{Replicator mutator space x}, \eqref{Replicator mutator space} is assured by classical arguments (see \cite{AKK},  \cite{Ikeda}). 
Monte-Carlo simulations  do not require
substantial changes compared to the non-spatial case: the additional Brownian motion can be effectively simulated in a standard way.


\section{A Fokker-Plank equation for the probability density}\label{fokkerplank}
The stochastic process \eqref{Replicator mutator space x}, \eqref{Replicator mutator space} can be described in a deterministic way by means of two Kolmogorov integro-partial differential equations: the backward one, also known as Feynman-Kac equation, (related to expected value) and the forward one, also known as Fokker-Plank equation (related to the density).

With minor changes from \cite[Proposition 3.1]{ACNT15}, one easily sees that the infinitesimal generator of the process \eqref{Replicator mutator space x} (settled in $\Sigma^d$), \eqref{Replicator mutator space} is
\begin{equation}\label{backward}
{\mathcal L} \phi=  a \cdot D_x \phi + v\cdot D_y \phi + \frac{1}{2}\mathrm{Tr} \left(\sigma\sigma^t D^2_{yy} \phi \right) + {\mathcal J} \phi .\end{equation}
Here $D_x$ and $D_y$ stand for the vectors of first derivatives w.r.t.~$x\in\R^{d-1}$ and $y\in\R^N$, respectively, $D^2_{yy}$ stands for the $N\times N$ matrix of the second order derivatives w.r.t.~$y$, $a$, $v$, $\sigma$ are the same functions appearing in \eqref{Replicator mutator space x}, \eqref{Replicator mutator space}, and $\jcal$ is a non-local functional related to a discrete measure:
\begin{align*}
    {\mathcal J} (x,y,\phi) = & \int_{\R^{d-1}} \left(\phi(x+z,y,t)- \phi(x,y,t) \right) d\mu_{x,y}(z) , \\
    \mu_{x,y} (z)= & \sum\limits_{\substack{ i, j = 1 \\ i\neq j}}^{d-1}\lambda_{ij} f_i(x,y) \delta_{\{\gamma_{ij}x_i ({e_j}-{e_i})\}}(z)  +  \sum\limits_{i=1}^{d-1}\lambda_{id} f_i(x,y) \delta_{\{-\gamma_{id}x_i{e_i}\}}(z) \\
    & + \sum\limits_{i=1}^{d-1}\lambda_{di} f_d(x,y) \delta_{\{\gamma_{di}(1-\sum\limits_{k=1}^{ d-1}x_k)\}}(z).
\end{align*}
The expected value at time $t$ of a population which is at state $(x,y)$ at time $t=0$ is described by  $u(x,y,t)$, the solution to the Feynman-Kac system
\begin{equation}\label{FeynmanKac}
\begin{cases}
\partial_t u_k - a\cdot D_x u_k -v\cdot D_y u_k - \frac{1}{2}\mathrm{Tr} \left(\sigma\sigma^t D^2_{yy} u_k \right)  ={\mathcal J} u_k , 
\\
u_k(x,y,0)=\begin{cases} x_k & \mbox{ as } k=1,\dots d-1 , \\ y_{k-d} & \mbox{ as } k=d,\dots d+N-1 . \end{cases}
\end{cases}
\end{equation}
Otherwise, one can be interested into the macroscopic function $\varrho(x,y,t)\in [0,1]$, measuring the probability of finding a population distribution $(x_1,\dots x_{d-1}, 1-\sum\limits_{i=1}^d x_i)\in S^d$ in the position $y\in\R^N$ at time $t$.
For instance at time $t>0$ the quantity
\[ P_i(t)=\iint\limits_{(B_{\varepsilon}(e_i)\cap \Sigma^d)\times\R^N} \varrho(x,y,t) dx dy \]
depicts the probability of having an high proportion of individuals of type $i$, while 
\[ P_i(t, \delta)=\iint\limits_{(B_{\varepsilon}(e_i)\cap \Sigma^d)\times B_{\delta}(0)} \varrho(x,y,t) dx dy \]
depicts the probability of finding an high proportion of individuals of type $i$ near at the origin.
\\
This can be done if the starting point is one population with character $x$ in the position $y$ (that is the initial datum is a Dirac function centered at $(x,y)$), but it becomes even more interesting if the initial status consists in a high numbers of subpopulations distributed according to a density $\varrho_0(x,y)$.
In this last case, the probability density $\varrho(x,y,t)$ furnishes a  statistical picture of a metapopulation composed by subpopulations which evolve according to selection and point-type mutations, and move independently.

To assure that the probability density is smooth, the classical theory by Hormander requests some technical assumptions, also in the diffusive setting (i.e.~in absence of mutations). 
We refer, for instance, to the book \cite{Nualart} for more details about the Hormander theory.
The investigation of a-priory regularity of the probability density greatly complicates in the presence of mutations.
We therefore choose to write the Fokker-Plank equation formally and then to settle it in the framework of viscosity solution theory. This approach has the advantage of asking very few a-priori regularity and producing well-posed solutions even in the integro-differential setting arising from rare mutations. 

Following Pavliotis \cite{Pavliotis} we  compute  $\mathcal{L}^{*}$, the dual operator in $L^{2}(\Sigma^d\times\R^N)$ of the infinitesimal generator:
\begin{equation}\label{forward} 
  {\mathcal L}^* \phi=  \frac{1}{2} \sum\limits_{h, k =1}^{N}\partial^{2}_{y_h y_k} \left((\sigma\sigma^t)_{h k }  \phi\right) -  {\mathrm{div}}_x \left(\phi a\right) - {\mathrm{div}}_y \left(\phi v\right) 
  + \sum\limits_{i=1}^d{\jcal}_i^*(f_i\phi) ,
\end{equation}
where now
\begin{align*}
  {\jcal}_i^*(x,y,\phi)= & \int_{\R^{d-1}} \left(\phi(x+z,y,t)- \phi(x,y,t) \right) d\mu^i_{x,y}(z) , \\
  d\mu^i_{x,y}(z) = & \sum\limits_{\substack{ j = 1 \\ j\neq i}}^{d-1}\lambda_{ij}(1+\gamma^*_{ij}) \, 1_{\Sigma^d}(x+ \gamma^*_{ij}x_i ({e_i}-{e_j})) \delta_{\{\gamma^*_{ij}x_i ({e_j}-{e_i})\}}(z) \\
  &  +  \lambda_{id}(1+\gamma^*_{id}) \, 1_{\Sigma^d}(x+ \gamma^*_{id}x_i{e_i}) \delta_{\{-\gamma^*_{id}x_i{e_i}\}}(z) ,
  \intertext{ as $i=1,\dots d-1$ and} 
d\mu^d_{x,y}(z) = & \sum\limits_{ j = 1}^{d-1}   \lambda_{dj}(1+\gamma^*_{dj}) \, 1_{\Sigma^d}(x-\gamma^*_{dj}(1-\sum\limits_{k=1}^{d-1} x_k){e_j})\delta_{\{\gamma^*_{dj}(1-\sum\limits_{k=1}^{d-1}x_k){e_j}\}}(z) , 
\end{align*}
for 
$\gamma^*_{ij} =\gamma_{ij}/(1-\gamma_{ij})$.
 It turns out that, if $\varrho_0(x,y)$ is the probability density of the random variable $X_0=(x_0,y_0)$ describing the initial distribution of subpopulations, and if the solution $X_t=(x_t,y_t)$ to \eqref{Replicator mutator space x}, \eqref{Replicator mutator space} has a sufficiently smooth probability density $\varrho(x,y,t)$ for $t>0$, then it solves the initial value problem
  \begin{equation}\label{eq:FP}
    \begin{cases}
      \partial_t\varrho  -\frac{1}{2}\sum\limits_{h , k  =1}^{N}\partial^{2}_{y_h y_k} \left((\sigma\sigma^t)_{ij}  \varrho\right)  +{\mathrm{div}}_x\left( \varrho a \right) +{\mathrm{div}}_y \left( \varrho v\right) 
      =\sum\limits_{i=1}^{d} \jcal_i^*(f_i\varrho) \\
      \varrho(x,y,0)=\varrho_0(x,y) ,\end{cases}
    \end{equation}
in the closed set $(x,y)\in \Sigma^d\times\R^N$ and $t>0$.

Let us explicitly remark that  nonlocal operators $\jcal^*_i$ are not continuous w.r.t.~$x$: this fact may have a huge instability effect.
We therefore switch to another problem   which is settled into all $\R^{d-1}\times\R^N$ and is continuous. To this end we extend the fitness functions $f_i$, the drift $v$ and the diffusion $\sigma$ in a bounded smooth way to all  $\R^d\times\R^N$  so that $f_i\ge 0$ have support  contained in a cylinder, say $B_R(0)\times\R^N$. Concerning the initial datum $\varrho_0$, it can be extended as $\varrho_0\equiv 0$ outside $\Sigma^d\times\R^N$.
  We thus look into the problem
  \begin{equation}\label{FP2}
    \begin{cases}
      \partial_t\varrho  -\frac{1}{2}\sum\limits_{h , k  =1}^{N}\partial^{2}_{y_h y_k} \left((\sigma\sigma^t)_{ij}  \varrho\right)  +{\mathrm{div}}_x\left( \varrho a \right) +{\mathrm{div}}_y \left( \varrho v\right) + c \varrho = \tilde\jcal\varrho  \\
    \varrho(x,y,0)=\varrho_0(x,y) ,\end{cases}
    \end{equation}
  for $(x,y)\in \R^d\times\R^N$ and $t>0$, where now
\begin{align*}
  \tilde{\jcal}(x,y,\phi)= & \int_{\R^{d-1}} \left(\phi(x+z,y,t)- \phi(x,y,t) \right) d\mu_{x,y}(z) , \\
  d\mu_{x,y}(z) = & \sum\limits_{\substack{i, j = 1 \\ j\neq i}}^{d-1} m_{ij}(x,y) \delta_{\{\gamma^*_{ij}x_i ({e_j}-{e_i})\}}(z) \\
  &  + \sum\limits_{i= 1}^{d-1} m_{id}(x,y)\delta_{\{-\gamma^*_{id}x_i{e_i}\}}(z) +  \sum\limits_{ j = 1}^{d-1}   m_{dj}(x,y)\delta_{\{\gamma^*_{dj}(1-\sum\limits_{k=1}^{d-1}x_k){e_j}\}}(z) , \\
 m_{ij}(x,y)= & (1+\gamma^*_{ij})\lambda_{ij} f_i(x+ \gamma^*_{ij}x_i ({e_i}-{e_j}),y),
\intertext{ as $i,j=1,\dots d-1$, with $i\neq j$, and}
m_{id}(x,y)= & \lambda_{id}(1+\gamma^*_{id})  f_i(x+ \gamma^*_{id}x_i,y), \\
 m_{di}(x,y)= &  \lambda_{di}(1+\gamma^*_{di})f_d(x-\gamma^*_{di}(1-\sum\limits_{k=1}^{d-1} x_k){e_i},y),
 \intertext{ as $i=1,\dots d-1$,} 
 c(x,y)=  & \sum\limits_{\substack{ i, j = 1 \\ i\neq j}}^{d} \left(\lambda_{ij} f_i(x,y) - m_{ij}(x,y)\right) .
\end{align*}
  
It is worth clarify that the equation in \eqref{FP2} does not coincide with the one in \eqref{eq:FP} even if $x\in \Sigma^d$. Although they do coincide for that functions  $\varrho$ which are zero for $x$ outside $\Sigma^d$.
On the other hand if the support of $\varrho_0$ is contained in $\Sigma^d\times\R^N$ and $\varrho(t)\in L^1(\R^{d-1}\times\R^N)$ is nonnegative, then also  the support of $\varrho(t)$ is contained in $\Sigma^d\times\R^N$.

To see this fact, let
\begin{align*}
  A_k=& \{x\in\R^{d-1} \, : x_k<0\}\quad & \text{ as } k=1,\dots d-1, \\
  A_d=& \{x\in\R^{d-1} \, : \sum\limits_{k=1}^{d-1}x_k >1\},  &  \\
  I_k(t)=&\iint\limits_{A_k\times\R^N}\varrho(t) dx dy \qquad & \text{ as } k=1,\dots d.
\end{align*}
It suffices to check  that  $\dfrac{d}{dt}I_k(t)\le 0$.
For simplicity we perform computations only in the case $d=2$.
 Integrating the equation in  \eqref{FP2} on $A_1\times\R^N$  gives
\begin{align*}
  \dfrac{d}{dt} I_1(t) =& -\int\limits_{\R^N} (a_1\varrho)(0,y)  dy
 +\lambda_{12}\iint\limits_{A_1\times\R^N}\left((1+\gamma^*_{12}) (f_1\varrho)((1+\gamma^*_{12})x, y,t) - (f_1\varrho)(x,y,t)\right) dx dy \\
 &  + \lambda_{21}\iint\limits_{A_1\times\R^N}\left((1+\gamma^*_{21}) (f_2\varrho)(x-\gamma^*_{21}(1-x),y,t)- (f_2\varrho)(x,y,t)\right) dx dy \\
 \intertext{remembering that $a_1(0,y)\equiv 0$ and performing the obvious transformations in the second and third integrals yields}
 = & -\lambda_{21}\int_{\R^N}dy\int_{-\gamma_{21}^*}^{0} dx (f_2\varrho)(x,y,t)   \le  0
\end{align*}
because $f_2\varrho\ge 0$.
Similarly, since $a_1(1,y)\equiv 0$ one gets
\begin{align*}
  \dfrac{d}{dt} I_2(t) =&  -\lambda_{12}\int_{\R}dy\int_1^{1+\gamma_{12}^*} dx (f_1\varrho)(x,y,t) \le   0 .
\end{align*}

It has also to be stressed that, in order to read the solution $\varrho(t)$ as a probability density, its total mass has to be $1$, that is 
\begin{align*}
M(t)=\iint\limits_{\R^{d-1}\times\R^N} \varrho(x,y,t) dx dy = 1 \quad & \mbox{ for all } t>0,
  \end{align*}
provided that $M(0)=\iint\limits_{\Sigma^d\times\R^N} \varrho_0(x,y) dx dy = 1$. 
Again, integrating the equation in  \eqref{FP2}   gives
\begin{align*}
  \dfrac{d}{dt} M(t) =&\lambda_{12}\iint\limits_{\R\times\R^N}\left((1+\gamma^*_{12}) (f_1\varrho)((1+\gamma^*_{12})x, y,t) - (f_1\varrho)(x,y,t)\right) dx dy \\
 &  + \lambda_{21}\iint\limits_{\R\times\R^N}\left((1+\gamma^*_{21}) (f_2\varrho)(x-\gamma^*_{21}(1-x),y,t)- (f_2\varrho)(x,y,t)\right) dx dy 
  = &0
\end{align*}
after a trivial change of variables.
Hence the total mass is preserved in the modified problem \eqref{FP2}.

In view of these remarks, we can read as the probability density of the process \eqref{Replicator mutator space x}, \eqref{Replicator mutator space} a solution $\varrho(t)$ to the Cauchy problem \eqref{FP2} with the properties $\varrho(t)\in L^1(\R^{d-1}\times\R^N)$ and $\varrho(t)\ge 0$ for  $t>0$.
The existence of such a solution is assured in the viscosity framework.

\begin{theorem}\label{teovisco}
  Assume that $f_i, v \in C^{1,1}(\R^{d-1}\times\R^N)$, $\sigma\in C^{2,1}(\R^{d-1}\times\R^N)$ are bounded together with their derivatives, with $f_i\ge 0$ and $\sigma\ge \ep>0$.
  Take $\varrho_0$ a Lipschitz-continuous, bounded function  whose support is compact and contained in the interior of $\Sigma^d\times\R^N$ such that $\varrho_0\ge 0$ and $\iint \varrho_0 dx dy =1$.
  Then there exists a unique viscosity solution to \eqref{FP2}.
  Moreover  $\varrho(t)\in L^1(\R^{d-1}\times\R^N)$ and $\varrho(t)\ge 0$ for all $t>0$.
  \end{theorem}
\begin{proof}
 First of all the equation in  \eqref{FP2} has to be written in the standard form of the viscosity solution framework, which is nonvariational. This can be done if the coefficients $f_i, v, \sigma$ have the regularity requested by hypothesis.   So we write
  \begin{equation}\label{FP3}
    \partial_t\varrho + a \partial_x \varrho + b \partial_y  \varrho + c \varrho -\frac{1}{2}  \sigma^2 \partial^2_{yy} \varrho = \sum\limits_{i=1}^2\tilde\lambda_{i} {\mathcal I}_i\left( \varrho \right)  \end{equation}
  where now
 \begin{align*}
\tilde \lambda_1(x,y)= & \lambda_{12}(1+\gamma_1^*) f_1 (x+\gamma_1^*x,y) , \\
\tilde \lambda_2(x,y)= & \lambda_{21}(1+\gamma_2^*) f_2 (x-\gamma^*_2(1-x),y) ,
\intertext{and consequently}
c(x) =  &  \partial_xa +\partial_y v -\frac{1}{2}\partial^2_{yy}\sigma^2 +
\sum\limits_{i=1}^2(\lambda_{i} f_i-\tilde\lambda_{i}) 
 \end{align*}
 are continuous and  bounded.
  This problem satisfies the assumptions in \cite[Theorems 1.1, 1.2]{Ama07}, therefore it has a unique continuous viscosity solution $\varrho(x,y,t)$ which is Lipschitz-continuous w.r.t.~$x,y$ and bounded.  Moreover comparison principle holds, in particular one can find suitable parameters $c_1,c_2,c_3$ so that
   \begin{equation}\label{integrabilita}
     0 \le \varrho \le \exp(c_1t - c_2\sqrt{1+x^2} - c_3 y^2) \quad \mbox{ in } \R^2\times[0,\infty) .\end{equation}
   In particular $\varrho(t)\in L^1(\R^2)$ for all $t$.
  \end{proof}

\begin{remark}
  The assumption $\sigma\ge \ep >0$ has only been used to obtain the estimate from above in  \eqref{integrabilita} and infer the integrability of the solution and the equation into all $\R^2$. The hypothesis can be removed by asking something more to the drift $v$ in order to assure some decay w.r.t.~$y$.
\end{remark}

Besides applications request to take a probability measure as  initial datum. This would hugely increase the mathematical difficulty. 
The recent paper \cite{KMS15} presents interesting results in this direction, which are modeled on the fractional Laplacian and therefore do not include the discrete non-local operator appearing here.
We also mention \cite{PR14} for some transport problem involving measures.

\section{Hawks and Doves: a numerical study}\label{Simulazioni}

In this section we take as a case study the two strategy game Hawks vs Doves ($d=2$),
with the following payoff matrix:
\[
\mathcal{U}=\left(\begin{array}{cc}
\frac{G-C}{2} & G\\
0 & \frac{G}{2}
\end{array}\right),
\]
where the coefficients are both positive. 
The fitness functions for Hawks ($x_{1}$) and Doves ($x_{2}$),
are respectively
\[
f_{1}=(G-C)x_{1}/2+Gy,\quad f_{2}={G}x_{2}/2,
\]
then the replicator dynamics (reducing the coordinates only to $x\in\left[0,1\right]$,
fraction of hawks) is
\[
\dot{x}=x\left(1-x\right)\left(f_{1}-f_{2}\right)=
x\left(1-x\right)\left(G-Cx\right)/2.
\]
Besides the pure-strategies equilibria $x=0$ (all Doves) and $x=1$
(all Hawks), a mixed strategies equilibrium can occur, $\bar{x}={G}/{C}$, when $C>G$: in this case you have the real
game of Hawk vs Doves, with $\bar{x}$ attractive and the other two values $0$ and $1$ which become unstable equilibria.
Notice that when the cost of the fight $C$ increases,
the percentage of hawks at the equilibrium $\bar{x}$ decreases; instead,
when the cost of fighting is less or equal than the gain, $C\leq G$,
the only equilibria in $\left[0,1\right]$ are the pure-strategies
ones, with $x=1$ attractive; the population tends to become only
hawks.

We add to the two strategies game also the space component, with $y\in\mathbb{R}$ ($N=1$). In particular we assume that the cost for fighting depends by $y$ as
\begin{align*}
  C(y)=& \frac{3G}{2}\left[1+ \frac{2}{\pi}\arctan(y)\right] 
  .
\end{align*}
The function is designed so that, at $y=0$, the cost for fighting is $C={3}G/{2}>G$ and we have a coexistence equilibrium $\bar{x}={2}/{3}$.
At the left of $y=0$ the cost decreases and it is equal to the gain for $y=-{\sqrt{3}}/{3}$, so  for smaller values of $y$ the coexistence equilibrium  disappears  and  hawks increase.
At the left of $y=0$  environment is more favorable to doves, because the cost increases up to $3G$, so that the fraction of hawks at equilibrium  decreases up to $\bar{x}={1}/{3}$.

The initial state  $\left(x,y\right)=({2}/{3},0)$ is an equilibrium when no motion is allowed, or when a deterministic motion  with $v({2}/{3},0)=0$ is considered. The following simulations show that this situation can be disrupted by Brownian motion and/or mutations.

Simulations that follow represent the density $\varrho(x,y,t)$, for $\left(x,y\right)\in[0,1]\times\mathbb{R}$,
and have been obtained in Matlab using Monte-Carlo methods. Roughly speaking, the stochastic
process is simulated for a high number of times, to statistically
estimate the density.
\begin{itemize}
\item Fixed the final time, $T$, we discretize the time interval $\left[0,T\right]$
in, at least, $N=2^{8}$ sub-intervals with the same length; fixed
an accuracy $\alpha$, the number $N$ increases up to make sure that
the probability of the event ``up to one jump in each interval''
is greater than $(1-\alpha)\%$;
\item We choose the number of iterations of the method, $\mathrm{itermax}$; we fix
two values, $N_{x}$, $N_{y}$ and the interval $\left[y_{min},y_{max}\right]$
in which we want to display the density; then we create a grid on
$[0,1]\times\left[y_{min},y_{max}\right]$, dividing the first interval
in $N_{x}$ parts, the second in $N_{y}$ ($y_{min}=-5$, $y_{max}=5$,
$N_{x}=N_{y}=50$); we define the array $H$ in three dimensions,
$N_{x}\times N_{y}\times N$, that will contain the following information:
\[
H(i,j,t)=\frac{\#\left\{ \mbox{processes s.t. at time \ensuremath{t} are in the cell grid \ensuremath{\left(i-1,i\right)\times\left(j-1,j\right)}}\right\}}{\mathrm{itermax}};
\]

\item For each iteration, we generate a Brownian motion on the $N$ time
points; then we generate a homogeneous Poisson process with intensity
$\lambda_{max}\geq\max_{x}\lambda(x)$ on $\left[0,T\right]$; let
$\left\{ T_{1},\ldots,T_{k}\right\} $ be the jump times;
\item We simulate, with Euler-Maruyama method, the stochastic process without
jumps, until the nearest time $T_{i}$;
\item Following the definition of the jump process and the intuitive interpretation
presented before, we decide (acceptance-rejection) if the jump of
the homogeneous process should be counted or not for the non-homogeneous
one: if not, we continue Euler-Maruyama until the next jump; if so,
we modify the population fractions in appropriate manner;
\item We update the array $H$.
\end{itemize}

\subsection{Replicator Dynamics perturbed by random motion.}
We take here that the population just moves randomly in space, subject to the selection of a changing environment. To do this, we imagine that jumps are absent, i.e.~$K=0$ in \eqref{Replicator mutator space x}, and that \eqref{Replicator mutator space} gives  an homogeneous Brownian motion for the variable $y$, i.e.~the drift $v$ is zero and the diffusion coefficient is $\sigma=0.2$.

If the Brownian motion were absent, the character of a population starting at $(x_0,y_0)$ would tend to
\[
x_{\infty}(y_0)=\begin{cases}
{G}/{C(y_0)} & y_{0}\geq-{\sqrt{3}}/{3} , \\
1 & y_{0}<{\sqrt{3}}/{3} .
\end{cases}
\]
as $t\rightarrow\infty$.
But now $y$ follows  an homogeneous Brownian motion, its marginal density  will
therefore be a Gaussian function with expected value $y_{0}$, kernel
of the heath equation,
\[
\varrho_{(y)}(y,t)=\frac{1}{\sqrt{2\pi\sigma^{2}t}}\exp\left\{ -\frac{\left(y-y_{0}\right)^{2}}{2\sigma^{2}t}\right\} .
\]
Meanwhile the SDE \eqref{Replicator mutator space x} reduces to the standard replicator dynamics and moves $x(t)$ towards
the asymptotically stable equilibrium, which in turns depends by $y$. We
can see how, with $t\gg0$, the density is approximately
\[
\varrho(x,y,t)\sim x_{\infty}(y)\varrho_{(y)}(y,t)
\]
with an approximately total number of hawks of $\int_{\mathbb{R}}x_{\infty}(y)\varrho_{(y)}(y,t)dy$.

\begin{figure}[htp]
\begin{center}

\begin{tabular}{cc}
\includegraphics[scale=0.45]{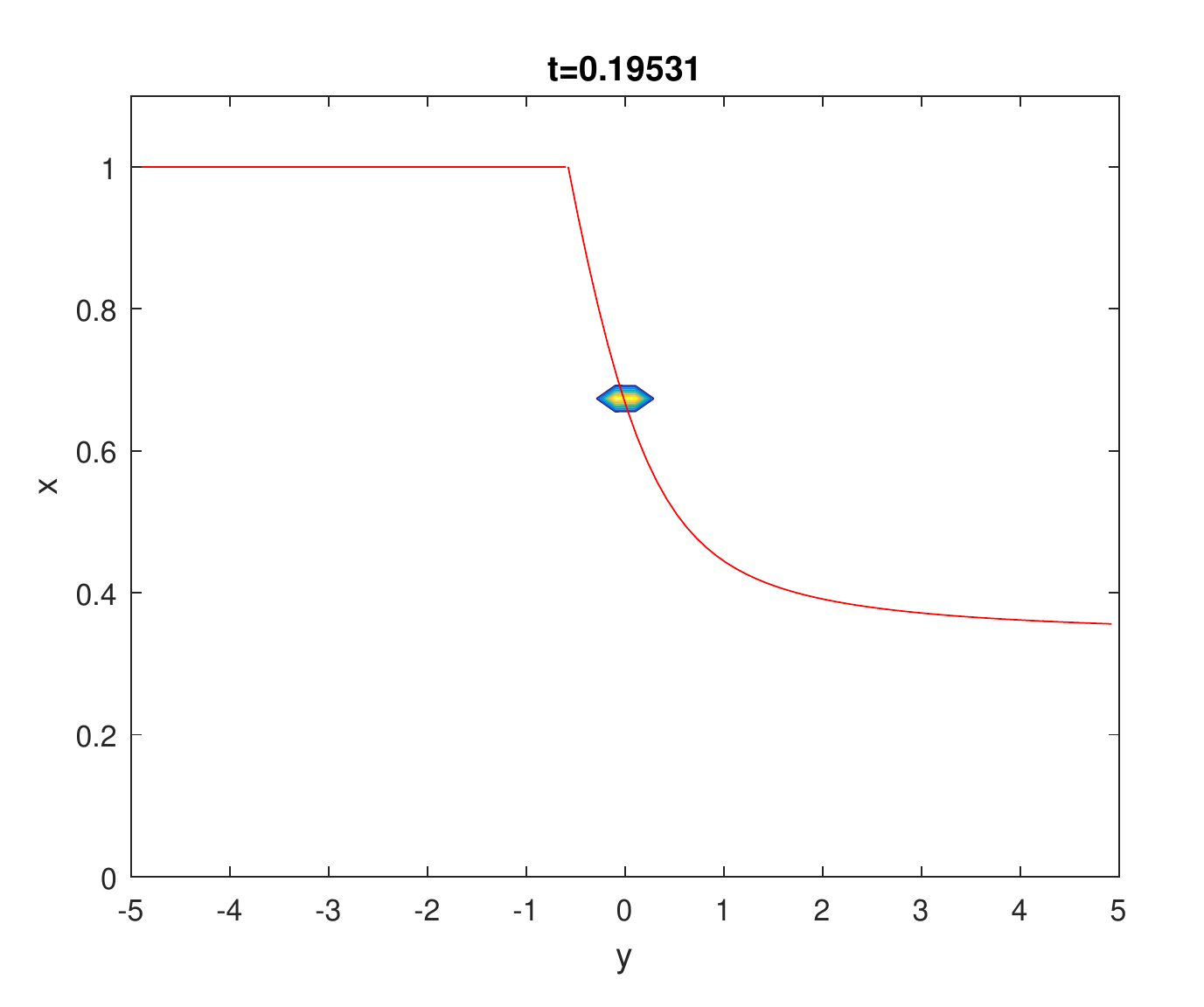}
&
\includegraphics[scale=0.45]{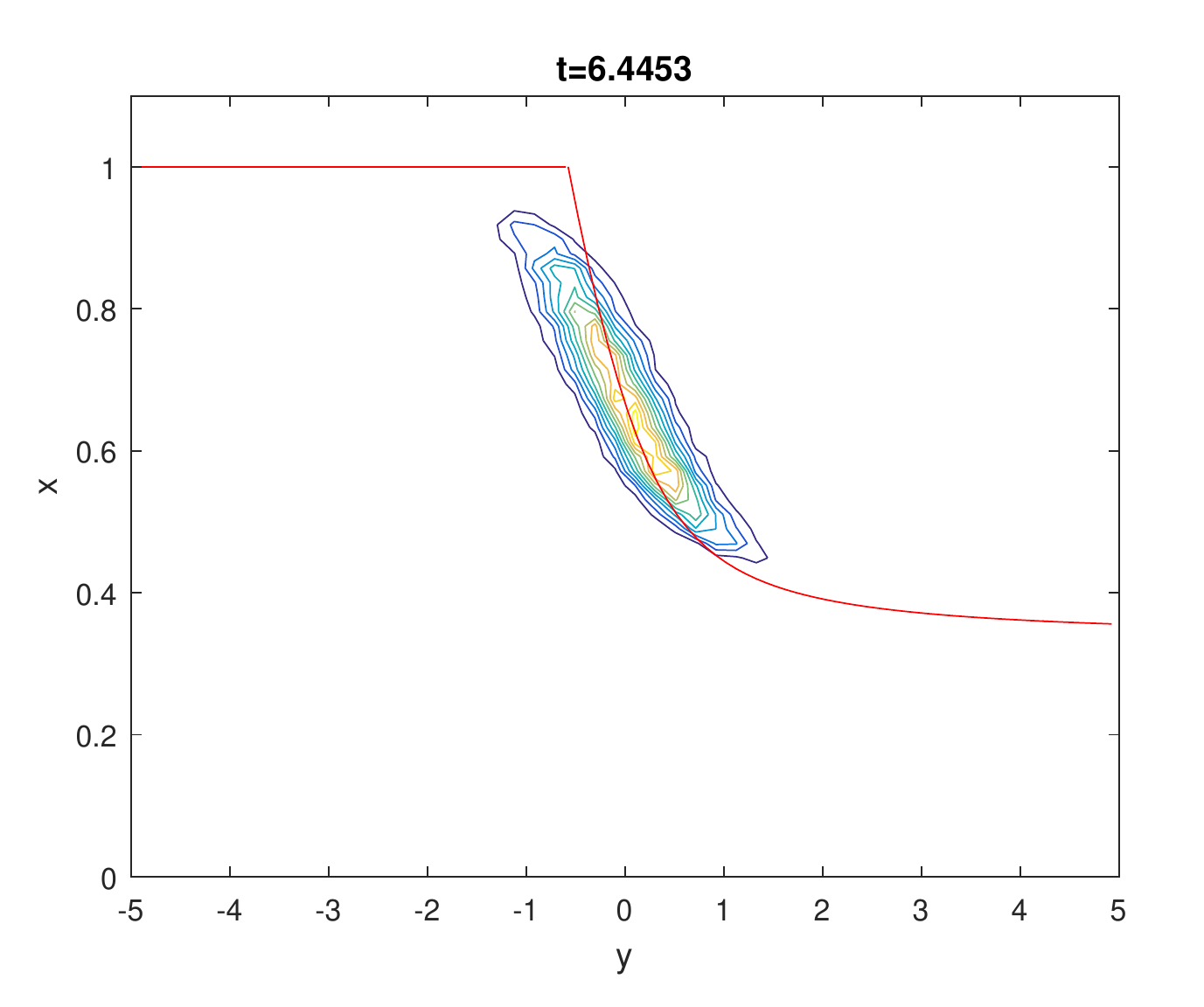}
\\
\includegraphics[scale=0.45]{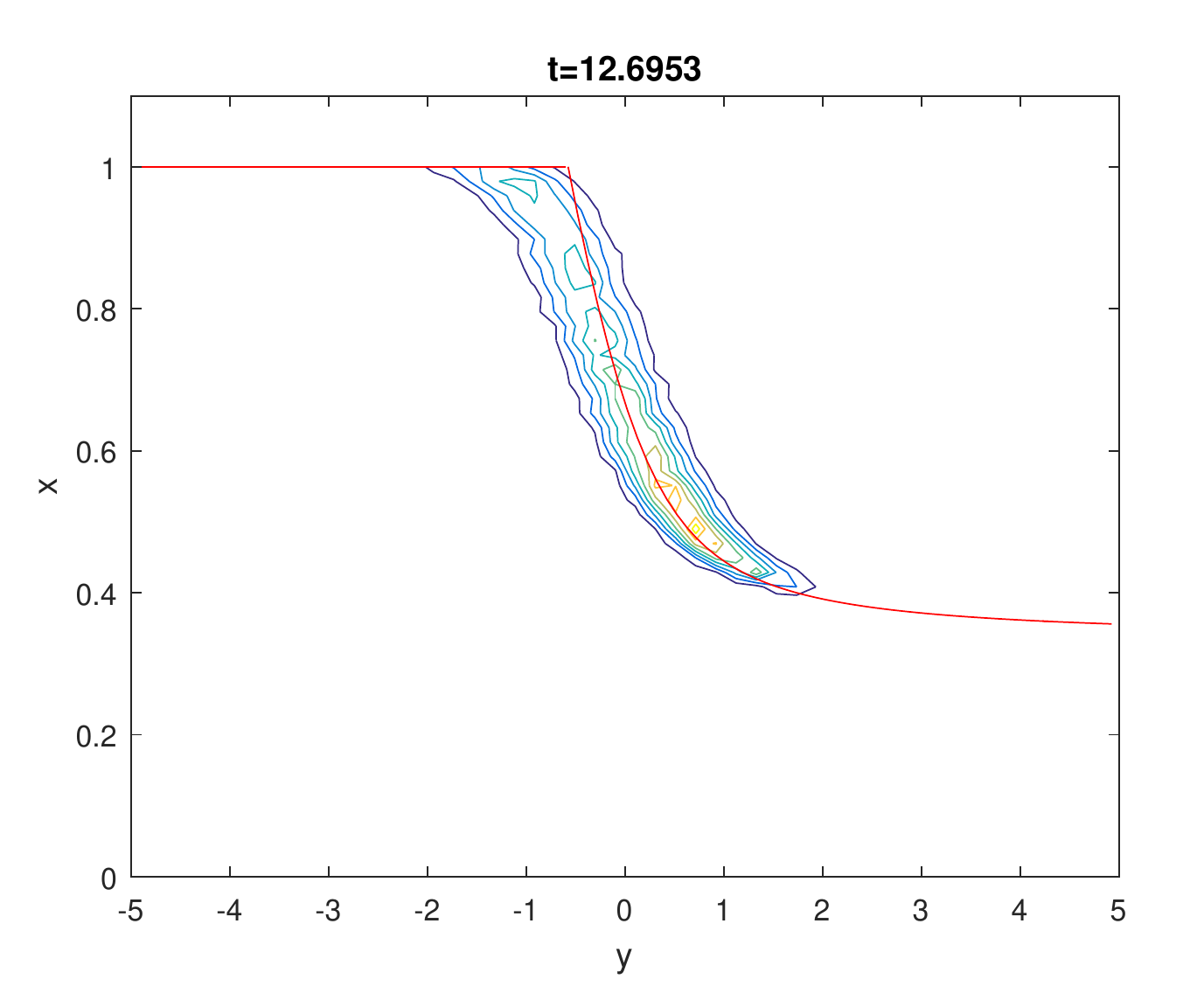}
&
\includegraphics[scale=0.45]{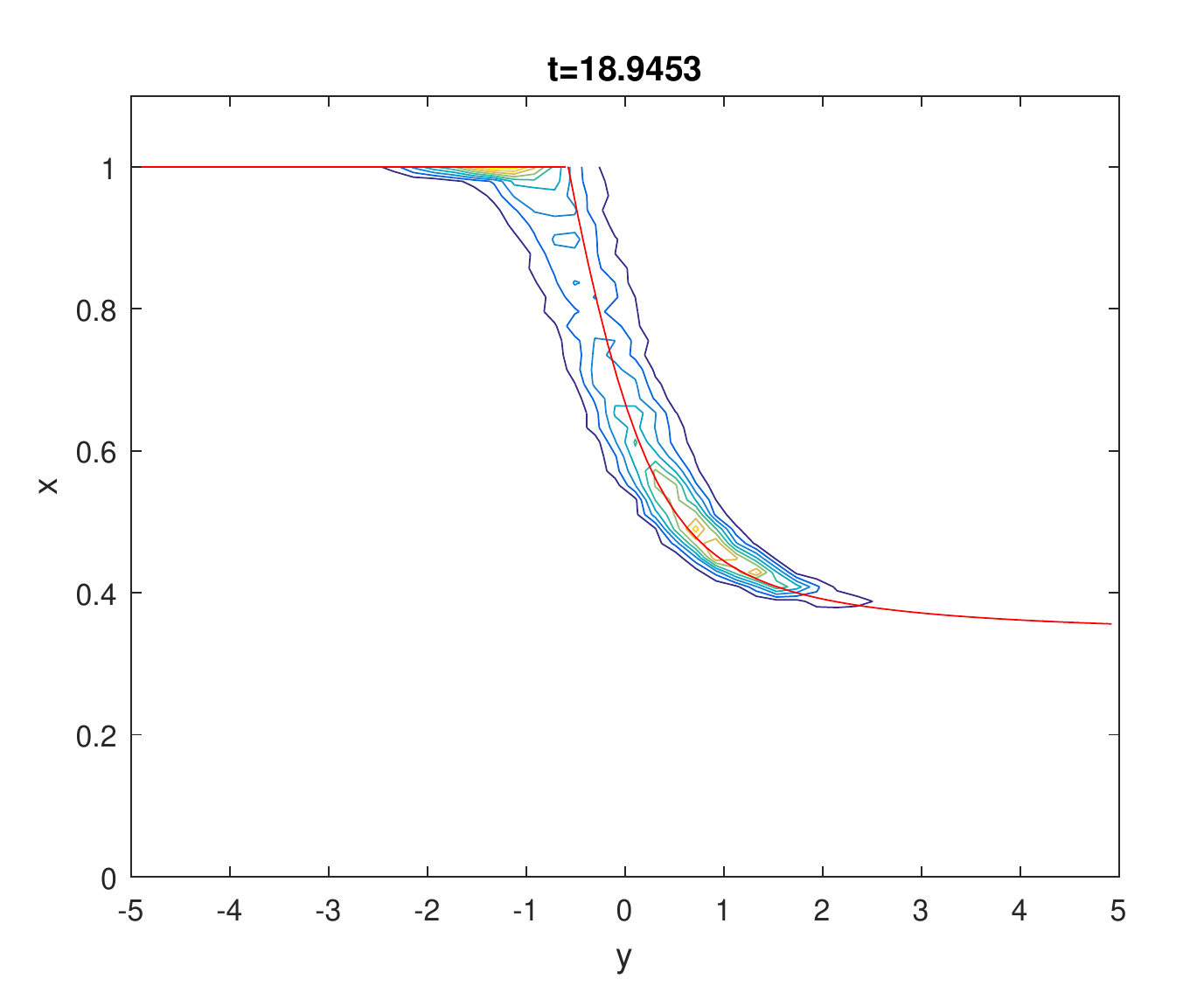}
\\
\includegraphics[scale=0.45]{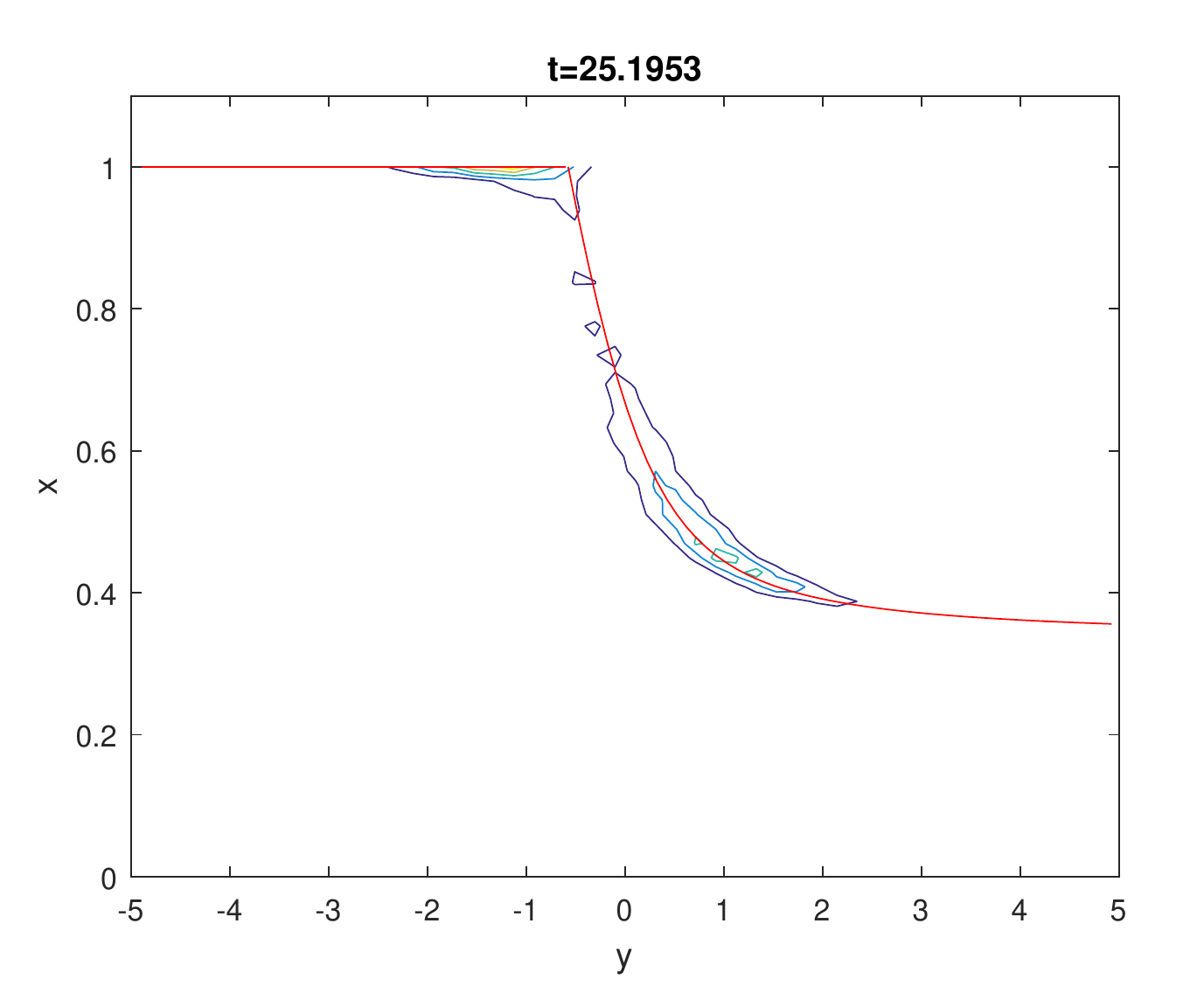}
&
\includegraphics[scale=0.45]{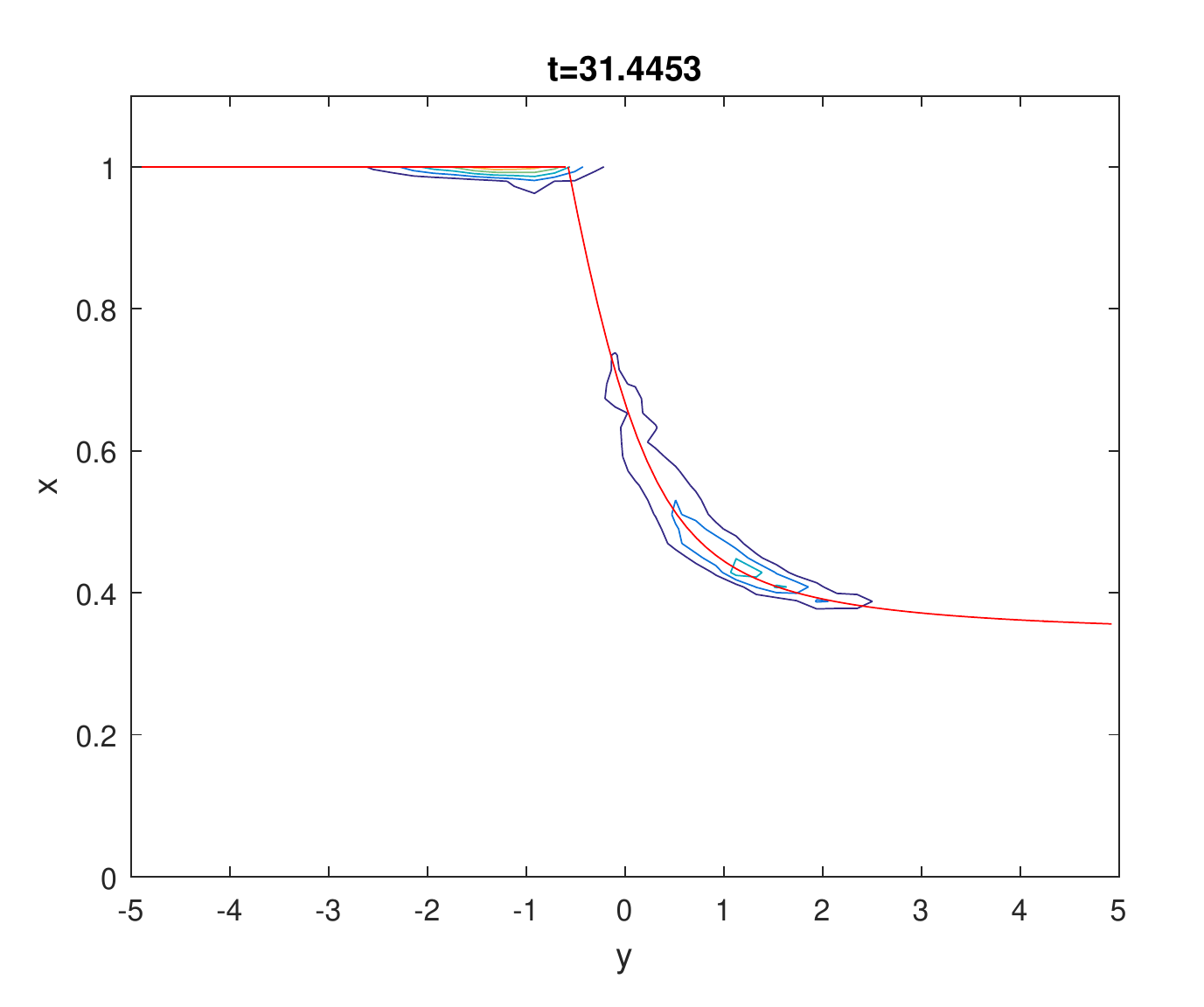}

\end{tabular}

\end{center}
\caption{Replicator Dynamics perturbed by random motion, simulated by  $10^5$ iterations. The population moves randomly in space, subject to the selection of a changing environment. The deterministic speed is zero, $v=0$, 
jumps are absent, the coefficient of the Brownian motion is $\sigma=0.2$. The red line is the function $x_{\infty}(y)$. Other parameters: $T=35$, $N=2^{8}$, $y_{min}=-5$, $y_{max}=5$, $N_{x}=N_{y}=50$.}

\end{figure}

\subsection{Replicator dynamics plus  Brownian motion with drift.}\label{rd+b}
We assume again that the character $x$ follows the replicator dynamics,  i.e. we take $K=0$ in \eqref{Replicator mutator space x}.
But now the position $y$ changes with nonnull drift, depending on the character of the population: 
\begin{equation}\label{drift}
v(x)=1-{3}x/2.
\end{equation}
The drift is decreasing as a function of $x$ (the proportion of hawks):
it has its maximum, $v=1$, at $x=0$
(Doves try to go right) and its minimum, $v=-{1}/{2}$, in $x=1$
(Hawks try to go left).
Moreover the drift is null at the equilibrium state $x={2}/{3}$, which is taken as the initial state; 
although the resulting system \eqref{Replicator mutator space x}, \eqref{Replicator mutator space} is not trivial because of the presence of the Brownian component with $\sigma=0.1$.

The presence of random motion leads to the formation of two
different regions: one group will continue to move leftward, the other rightward.
In the first one the proportion of hawks  will increase, until it overstep $y=-{\sqrt{3}}/{3}$,
after which the only new equilibrium will be $x=1$; then we will
see the gradual extinction of each dove.
In the other one however the cost of fight will increase with time,
the density will tend to concentrate 
toward the coexistence of both strategies, with greater concentration
of doves. 

As we can see numerically, the two regions have each mass ${1}/{2}$:
this means that the global fraction of hawks remains unchanged, although ``geographically'' the situation is very different from the beginning.

\begin{figure}[htp]
\begin{center}

\begin{tabular}{cc}
\includegraphics[scale=0.45]{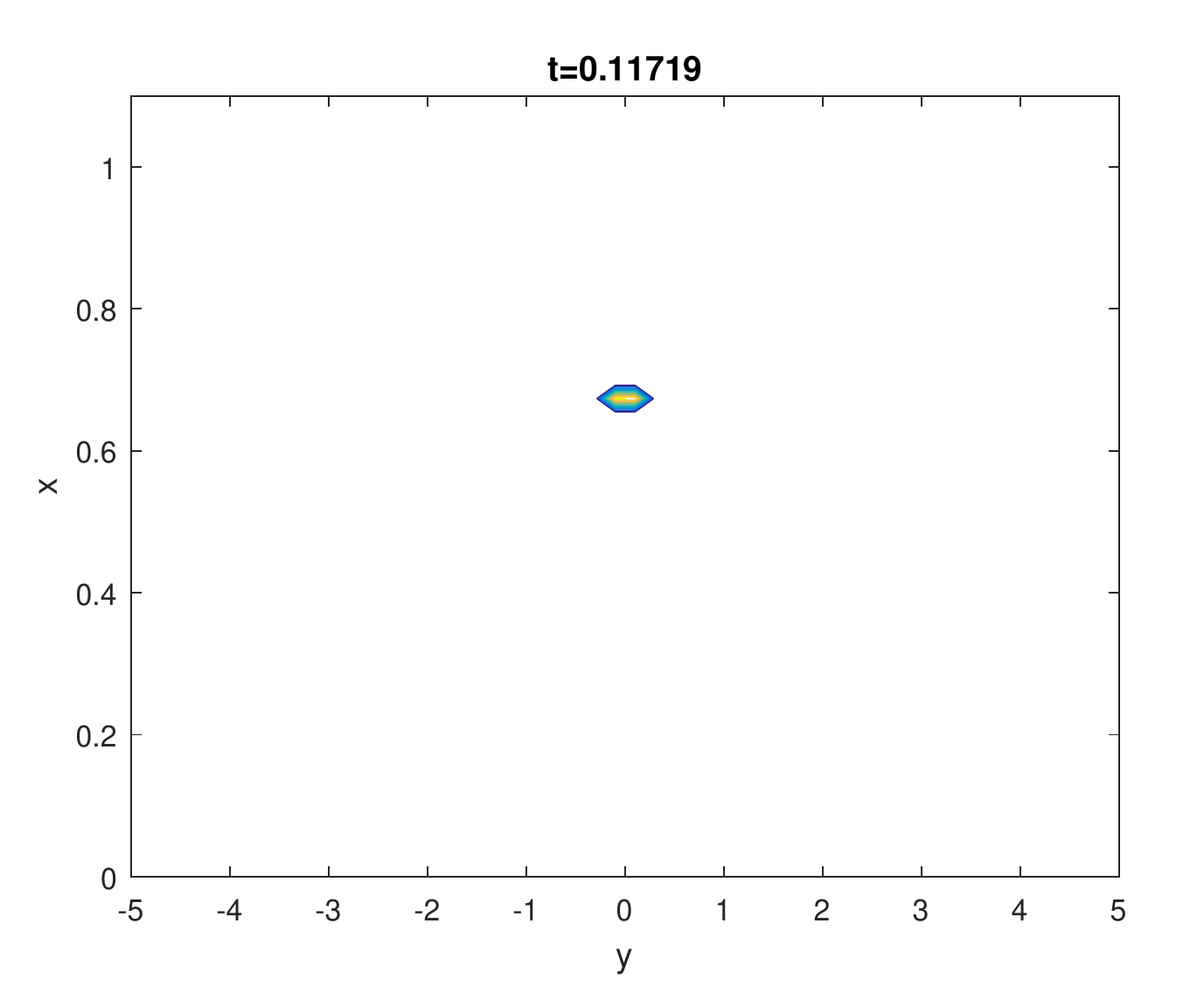}
&
\includegraphics[scale=0.45]{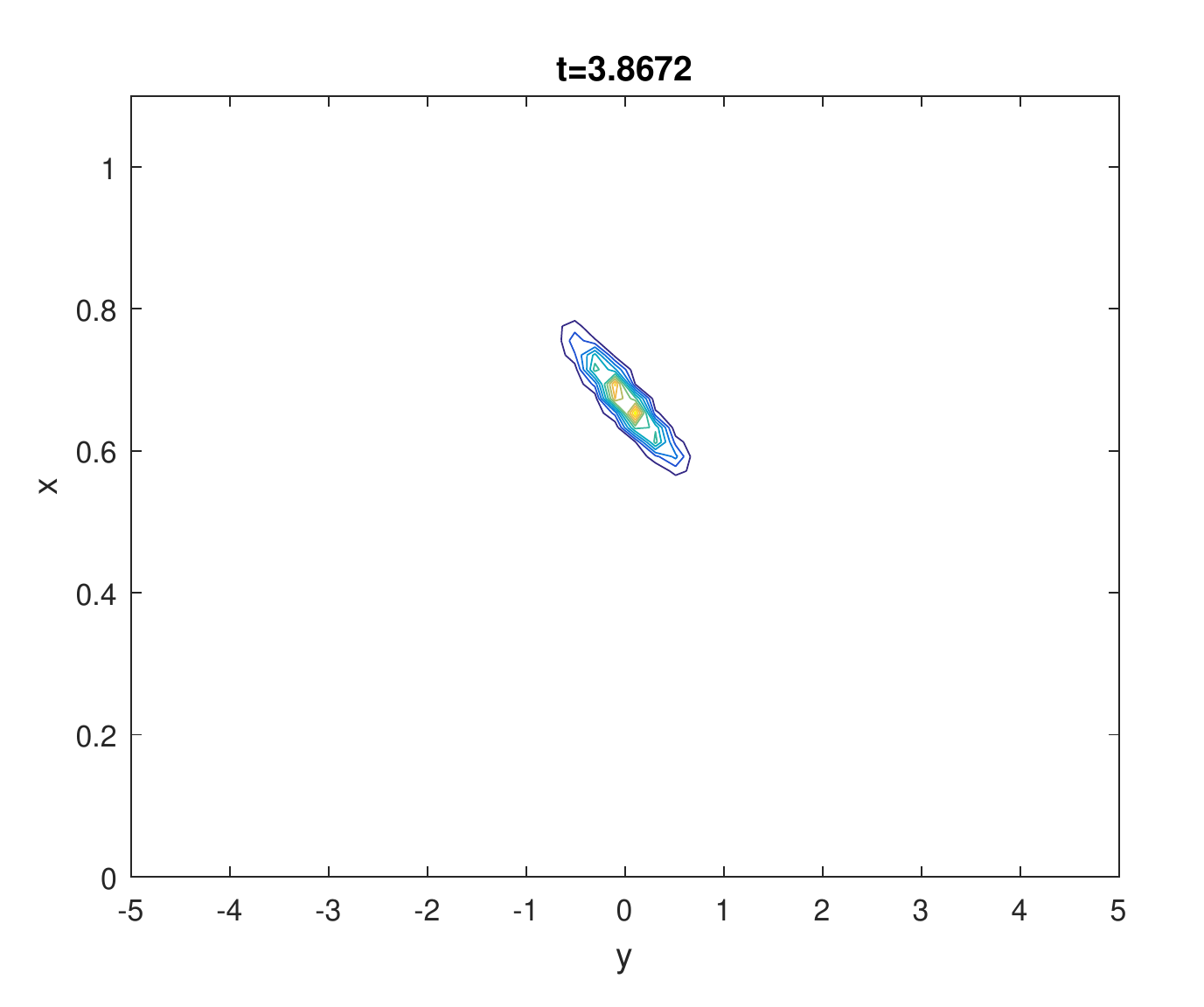}
\\
\includegraphics[scale=0.45]{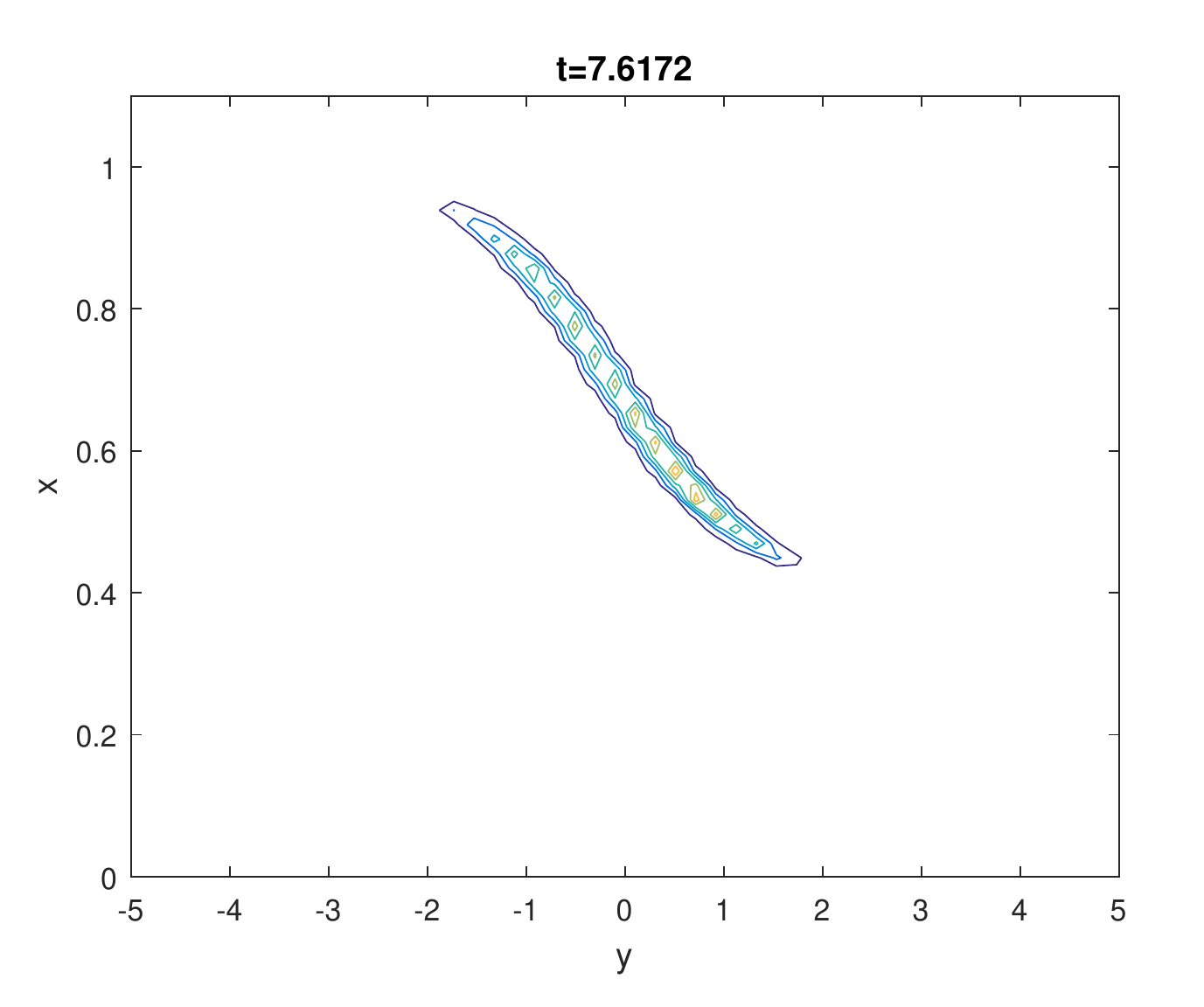}
&
\includegraphics[scale=0.45]{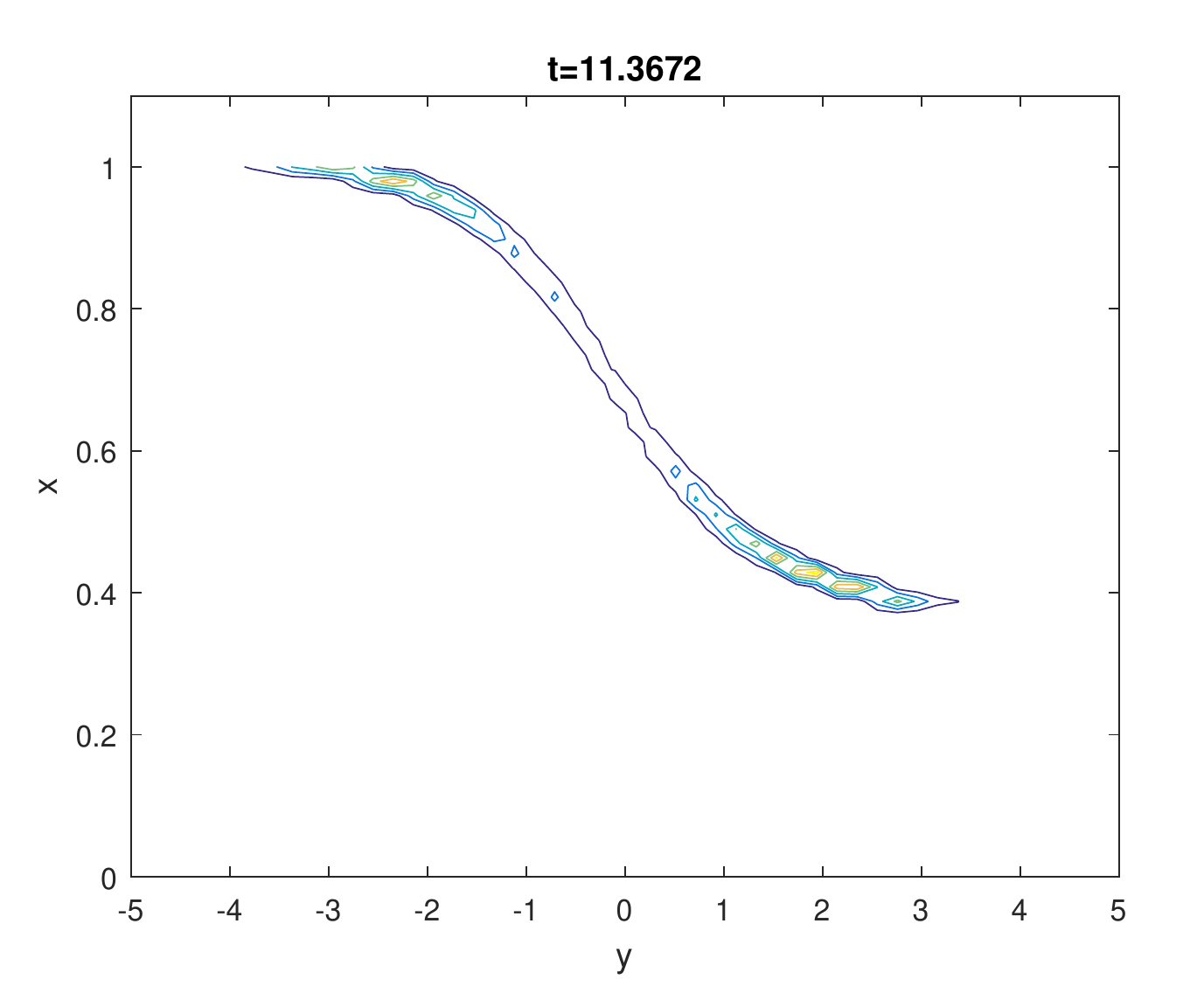}
\\
\includegraphics[scale=0.45]{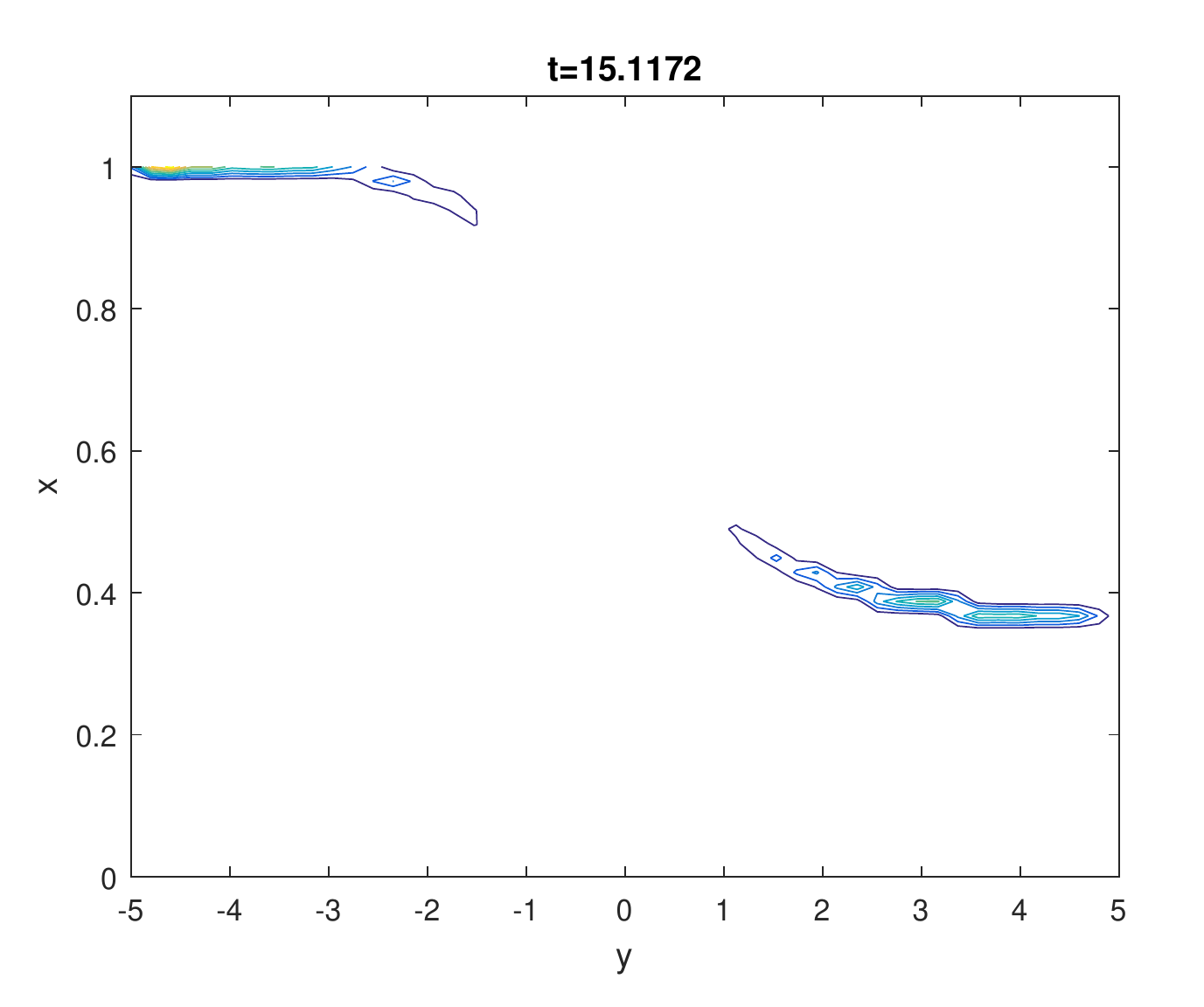}
&
\includegraphics[scale=0.45]{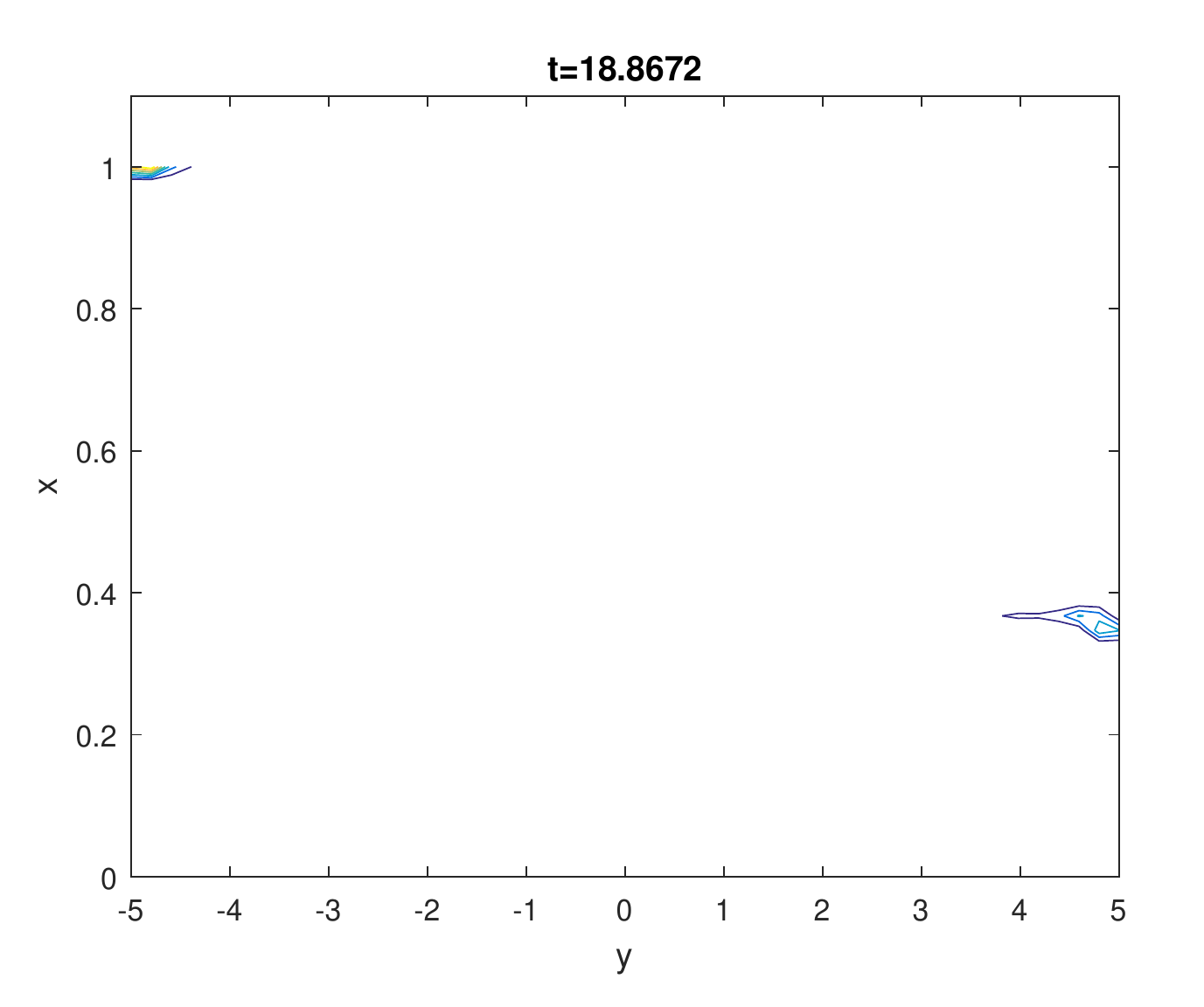}
\\
\includegraphics[scale=0.45]{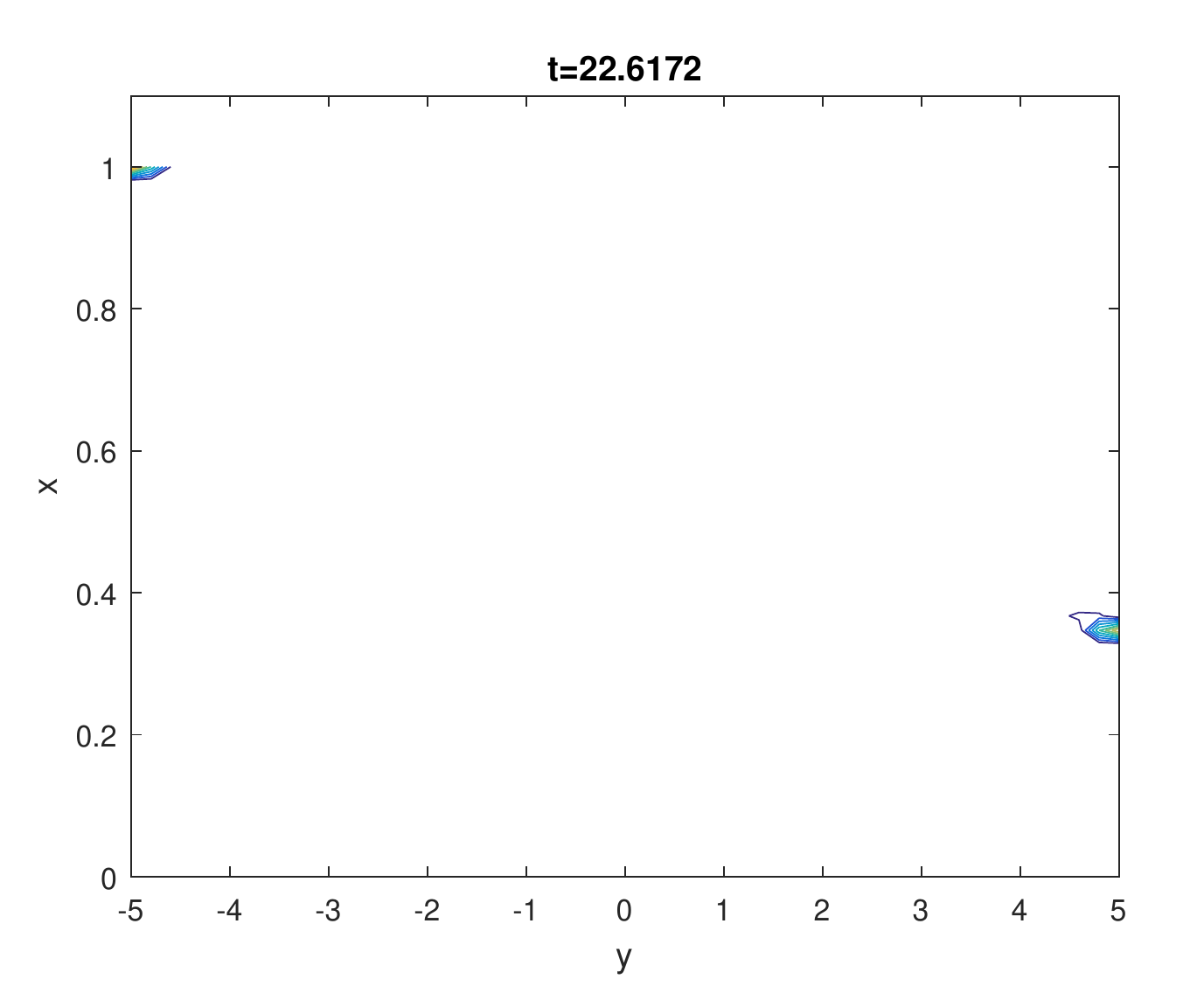}
&
\includegraphics[scale=0.45]{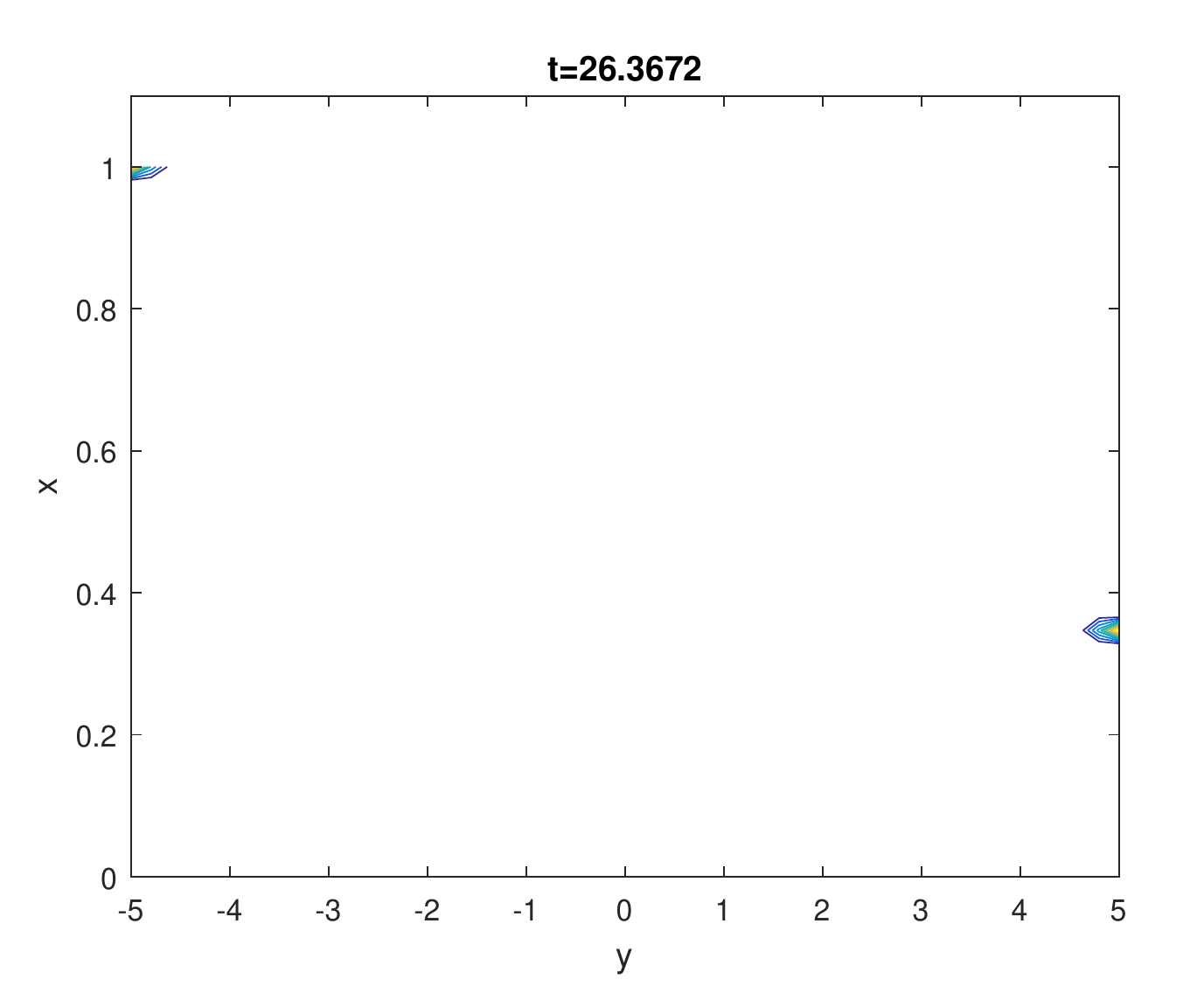}

\end{tabular}

\end{center}
\caption{Replicator dynamics plus  Brownian motion with drift, simulated by $10^5$ iterations. The vector $x$ evolves according to the game Hawks vs Doves with $G$  function of $y$. 
The deterministic speed is chosen as $v=1-\frac{3}{2}x$, jumps are absent, the coefficient of the Brownian motion is $\sigma=0.1$. Other parameters: $T=30$, $N=2^{8}$, $y_{min}=-5$, $y_{max}=5$, $N_{x}=N_{y}=50$.}

\end{figure}

\subsection{Point-type mutations plus deterministic motion}
We take a point-type mutation process for $x$, with $\lambda_{12}=\lambda_{21}=0.2;\,\gamma_{12}=\gamma_{21}=0.1$ in \eqref{K}.
Concerning motion, we take here $\sigma=0$, so that the position $y$ changes deterministically with speed $v$ given by \eqref{drift}.
Let us remark that at each time that a mutation occurs, the probability that hawks (respectively doves) suffer a mutation only depends by fitness.
Starting from the initial equilibrium $(2/3,0)$, the probability that hawks are the first to suffer mutations is $1/2$, just like doves.
In this sense mutations make up a random perturbation similar to the Brownian motion introduced in the previous example \ref{rd+b}. Simulations show two moving regions also in this case, indeed.
It is remarkable that the very fact that at the equilibrium  hawks are more abundant then doves brings as a consequence that mutations will favor doves, so that  the region moving rightwards will have higher mass and the total population of hawks will decrease, unlike  example \ref{rd+b}.
We therefore see that including the physical space can favour the persistence of the low-fitness strategy, when mutations can happen in both directions.
\begin{figure}[htp]
\begin{center}

\begin{tabular}{cc}
\includegraphics[scale=0.45]{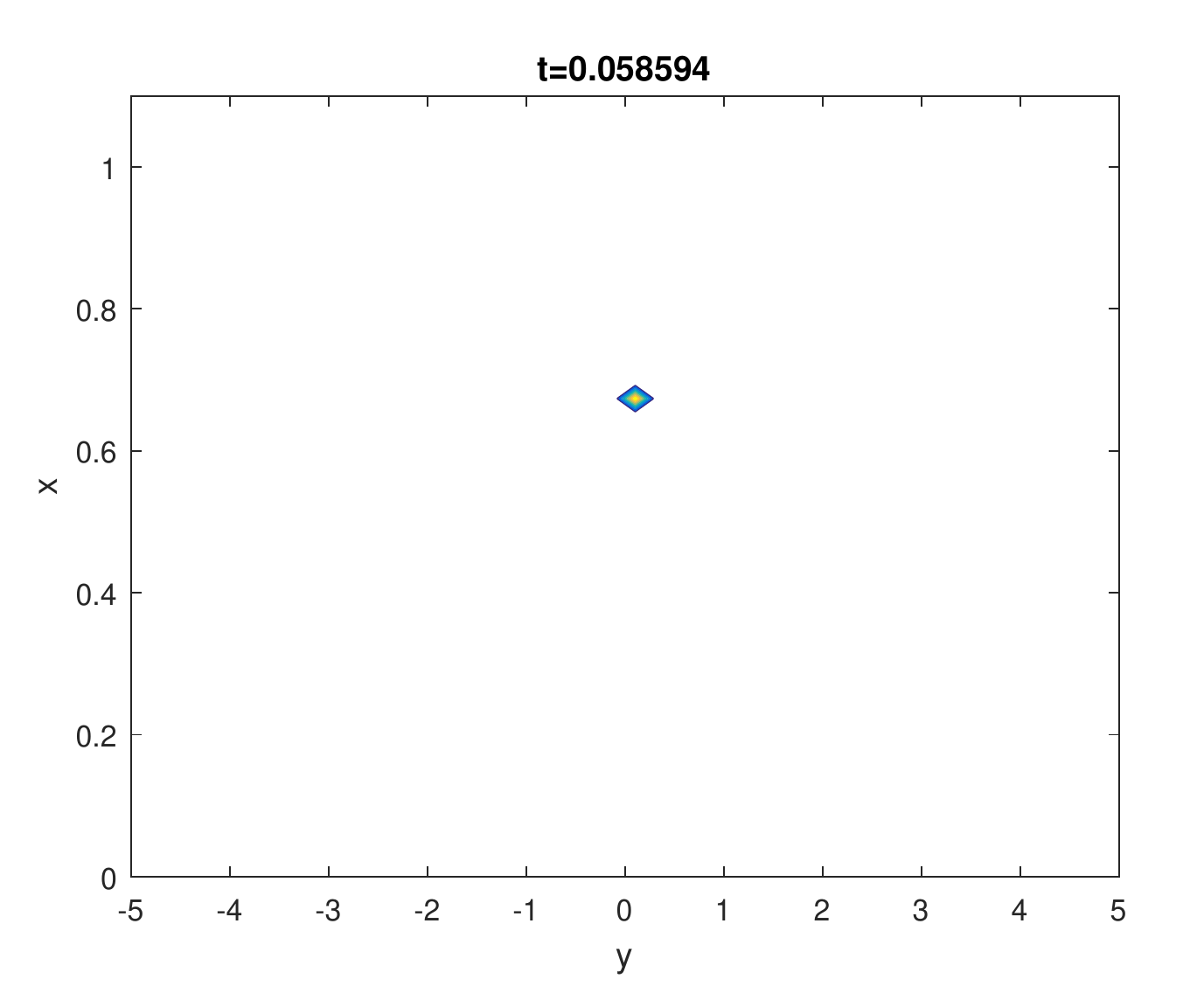}
&
\includegraphics[scale=0.45]{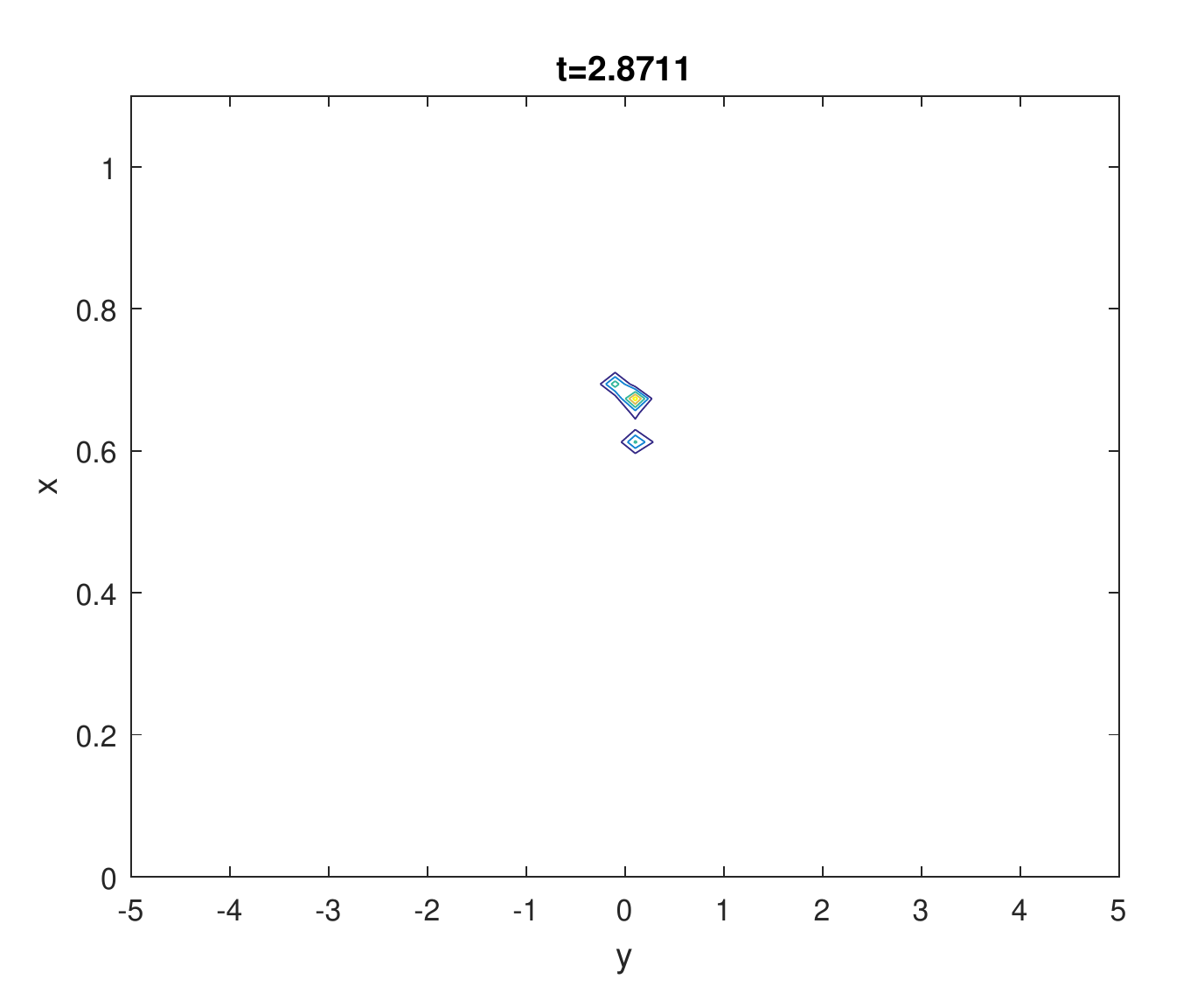}
\\
\includegraphics[scale=0.45]{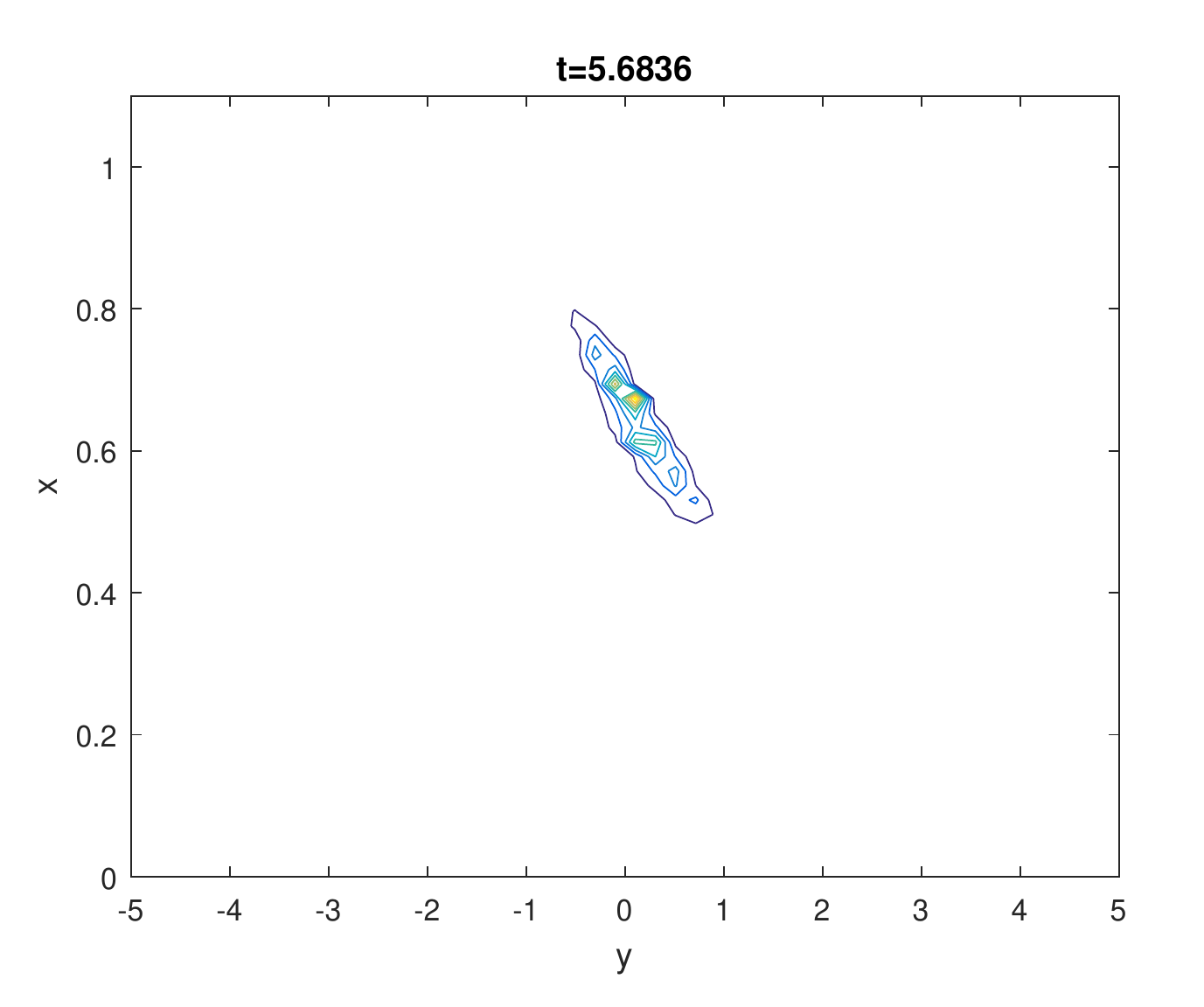}
&
\includegraphics[scale=0.45]{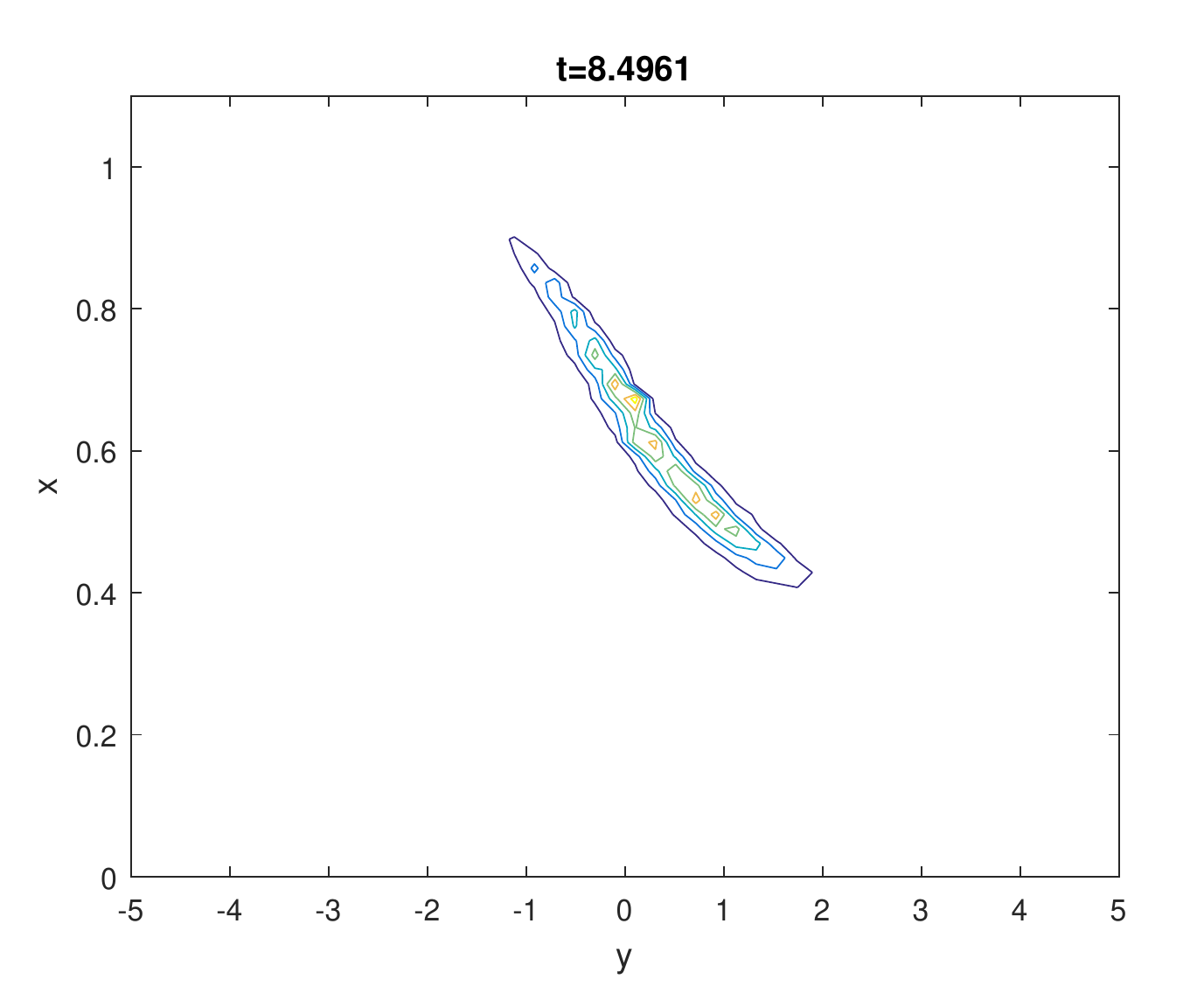}
\\
\includegraphics[scale=0.45]{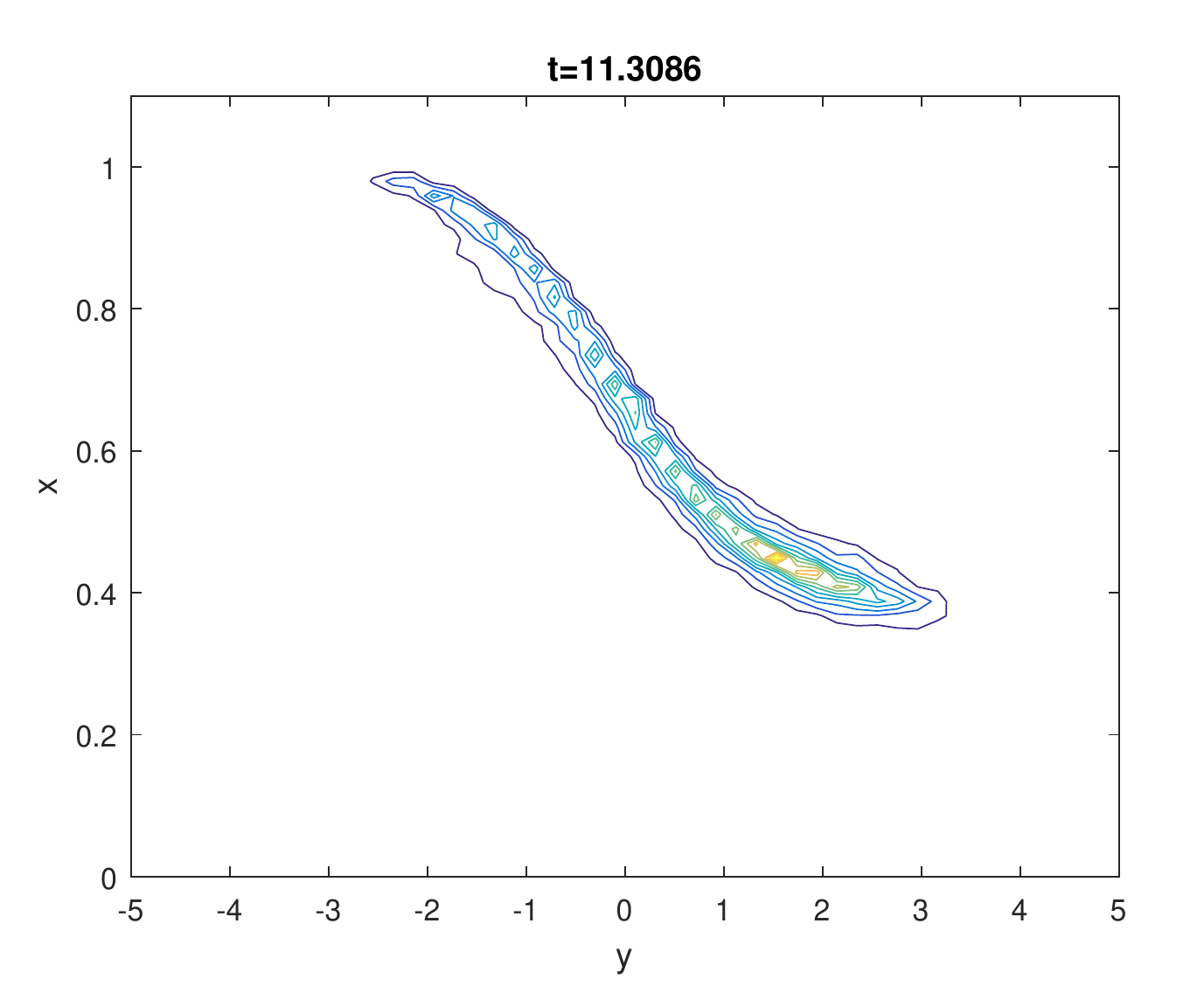}
&
\includegraphics[scale=0.45]{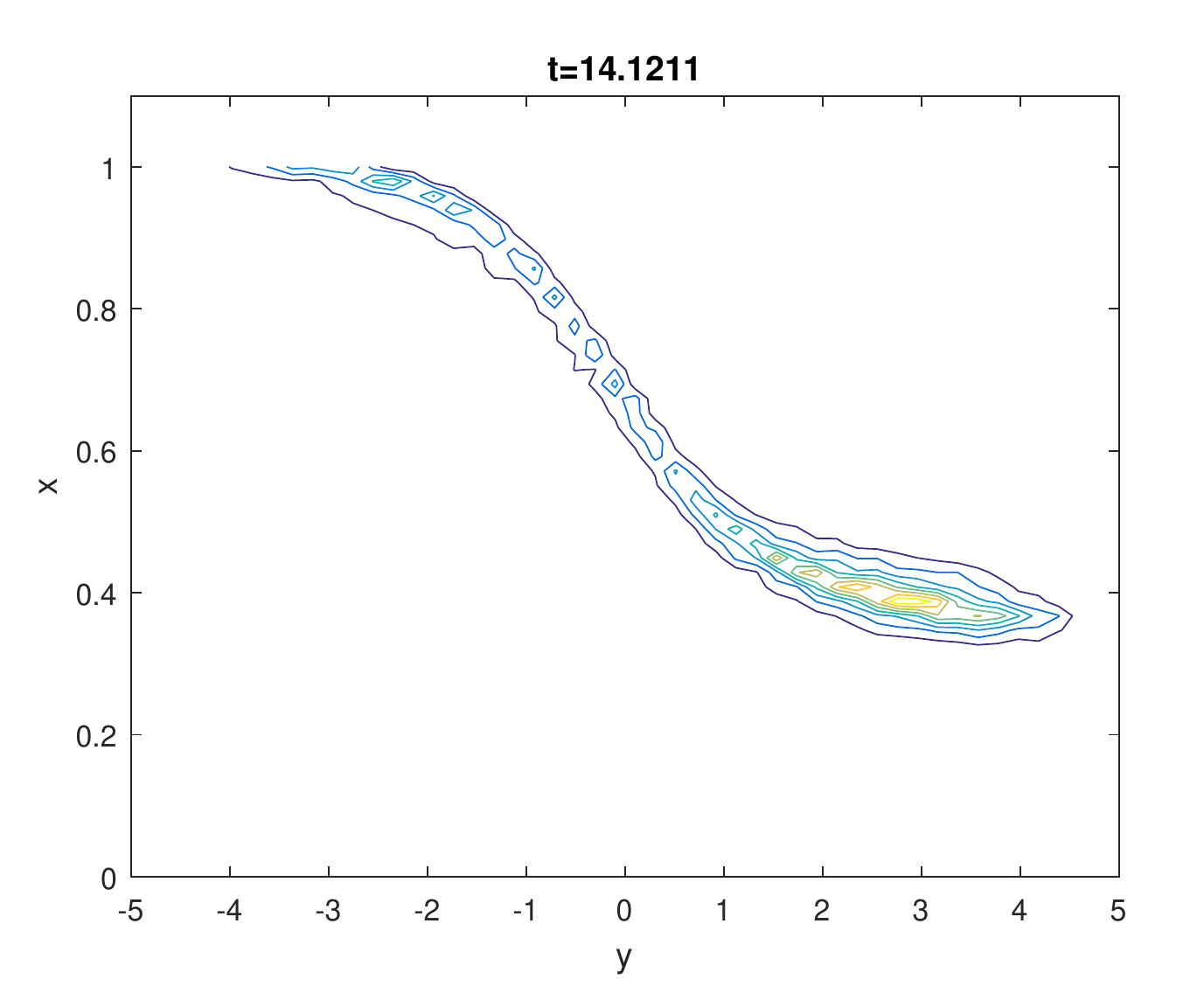}
\\
\includegraphics[scale=0.45]{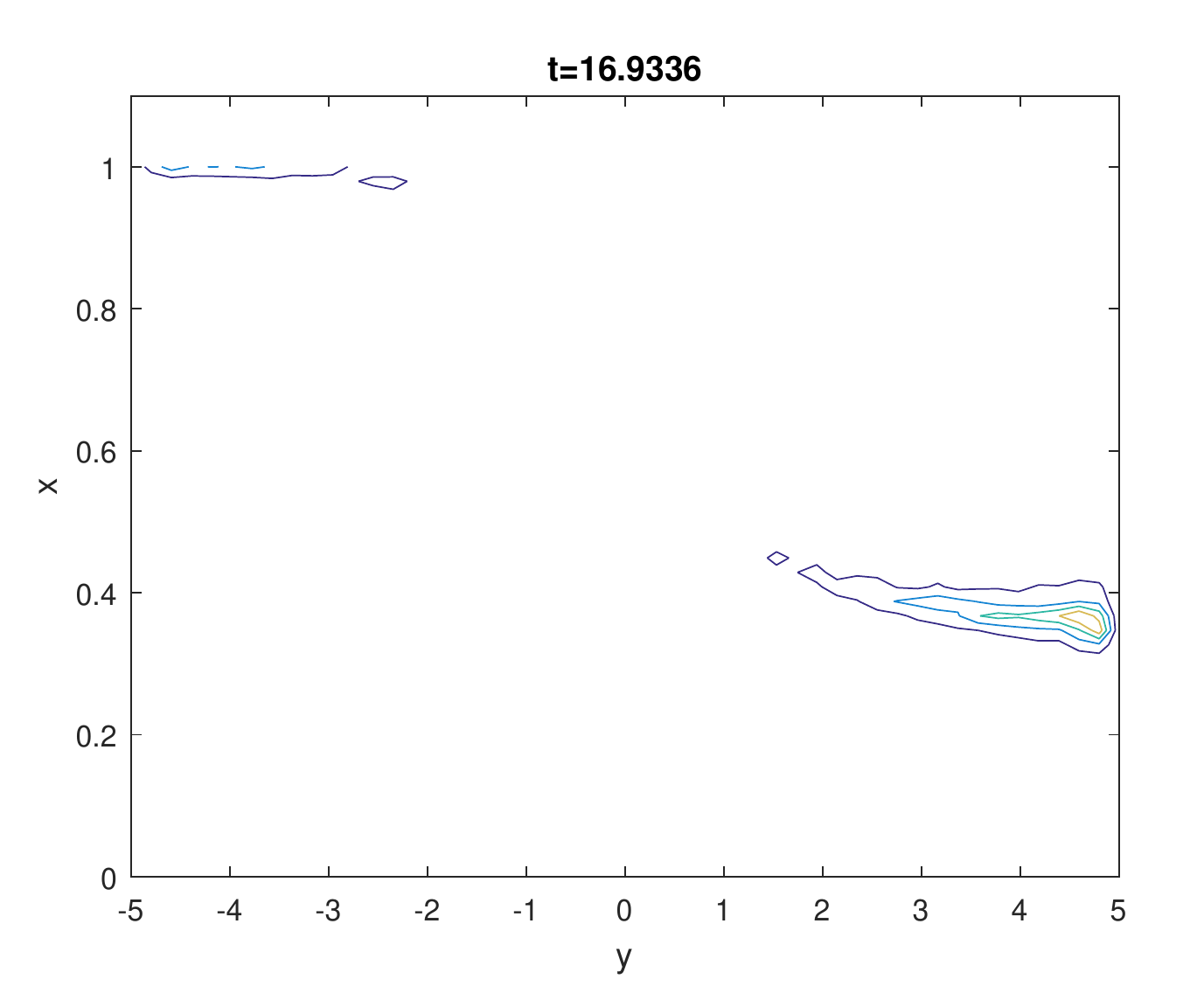}
&
\includegraphics[scale=0.45]{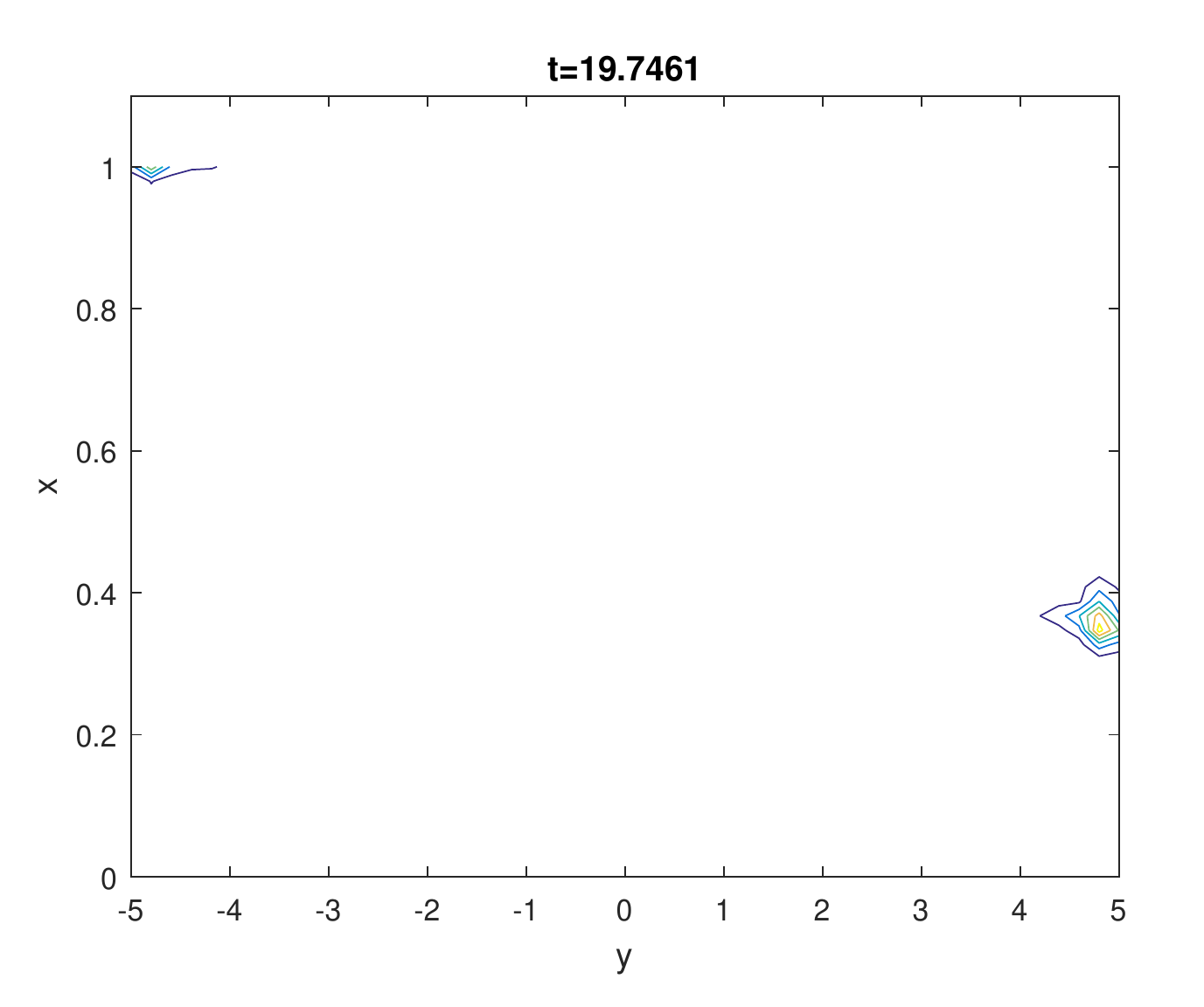}
\\
\end{tabular}

\end{center}
\caption{Point-type mutations plus deterministic motion, simulated by $10^5$ iterations. The vector $x$ evolves according to the game Hawks vs Doves with $G$ function of $y$. 
The deterministic speed is chosen as $v=1-\frac{3}{2}x$, Brownian motion is absent, the parameters of the jump process are $\lambda_{12}=\lambda_{21}=0.2$, $\gamma_{12}=\gamma_{21}=0.1$,
that is one tenth of the population mutate each jump and we have “fair jumps”. Other parameters: $T=30$, $N=2^8$, $\alpha=0.1$, $y_{min}=-5$, $y_{max}=5$, $N_{x}=N_{y}=50$.}
\end{figure}

\subsection{Point-type mutations plus Brownian motion with drift.}
In this last simulation, we choose as parameters $\sigma =0.1$, $\lambda_{12}=\lambda_{21}=0.2,\,\gamma_{12}=\gamma_{21}=0.1$. We do not assist in a substantial change of the dynamic, which looks similar to the two previous cases, with the formation of two regions of different mass. We then tested the convergence of the numerical method by changing the number of samples, doubling them in each of 11 simulations, using 100, 200, 400, 800, 1600, 3200, 6400, 12800, 25600, 51200, 102400 points. We define the error $e_{n}$ as
\begin{equation}
e_{n}=\left|\varrho^{(2n)}\left(\cdot,\cdot,\frac{T}{2}\right)-\varrho^{(n)}\left(\cdot,\cdot,\frac{T}{2}\right)\right|_{1},
\end{equation}
and we suppose there exist $c\in\R$ and $\gamma>0$ such that, definitely, $e_{n}=cn^{-\gamma}$. We then estimate the order $\gamma$ with $\gamma=\log_{2}\left(\frac{e_{n}}{e_{2n}}\right)$. 
The approximation of $\varrho(\cdot,\cdot,\frac{T}{2})$, named $H(\cdot,\cdot,\frac{T}{2})$ in the program, is an array defined on grid $G$ made of $N_{x}\times N_{y}$ rectangles of area $A=\frac{(y_{max}-y_{min})}{N_{x}N_{y}}$;  
then we can calculate the error 
\begin{equation}
e_{n}\sim A\cdot\sum_{i,j}\left|H^{(2n)}(i,j,\frac{T}{2})-H^{(n)}(i,j,\frac{T}{2})\right|.
\end{equation}
Results are shown in figure.

\begin{figure}[htp]
\begin{center}

\begin{tabular}{cc}
\includegraphics[scale=0.45]{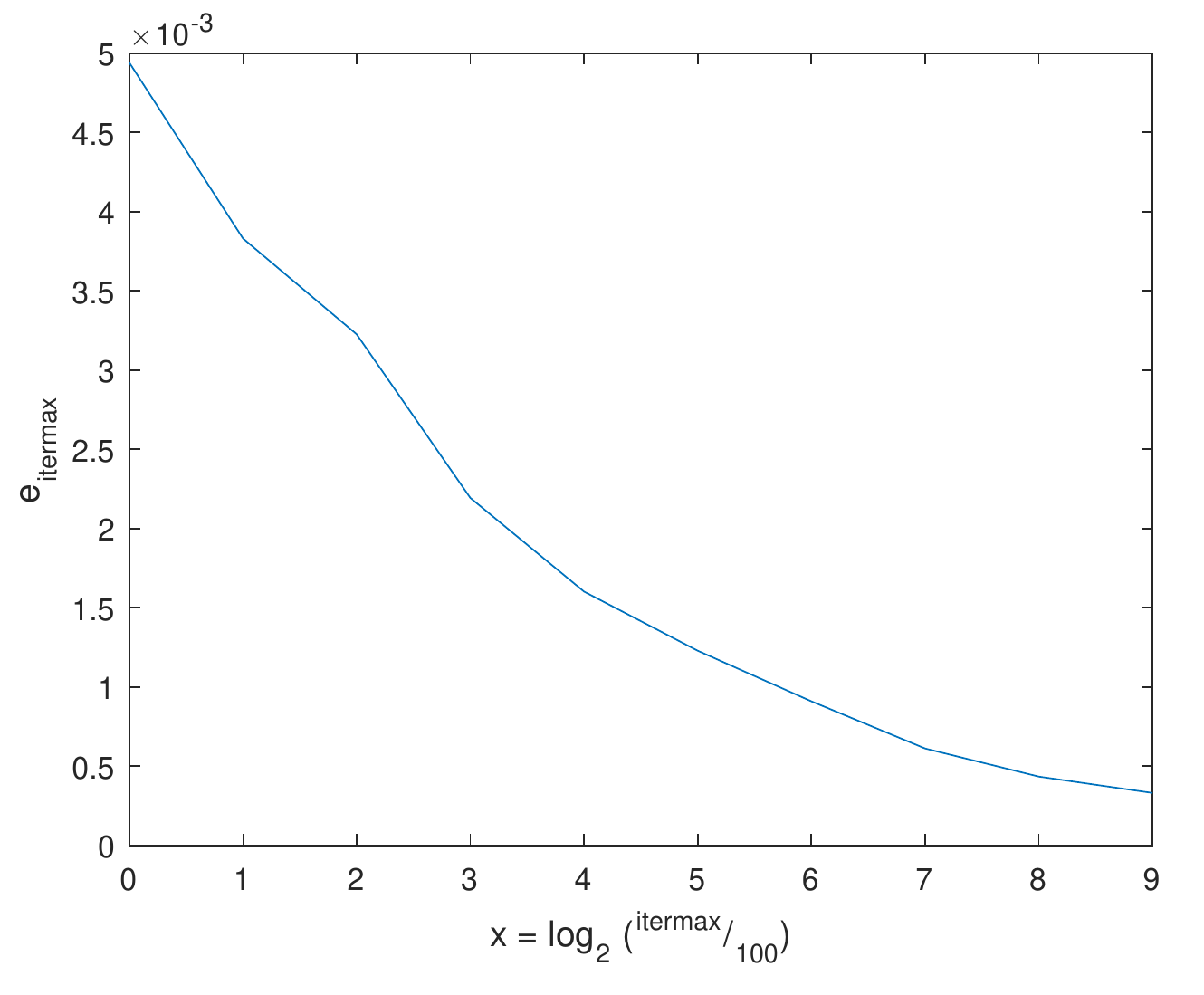}
&
\includegraphics[scale=0.45]{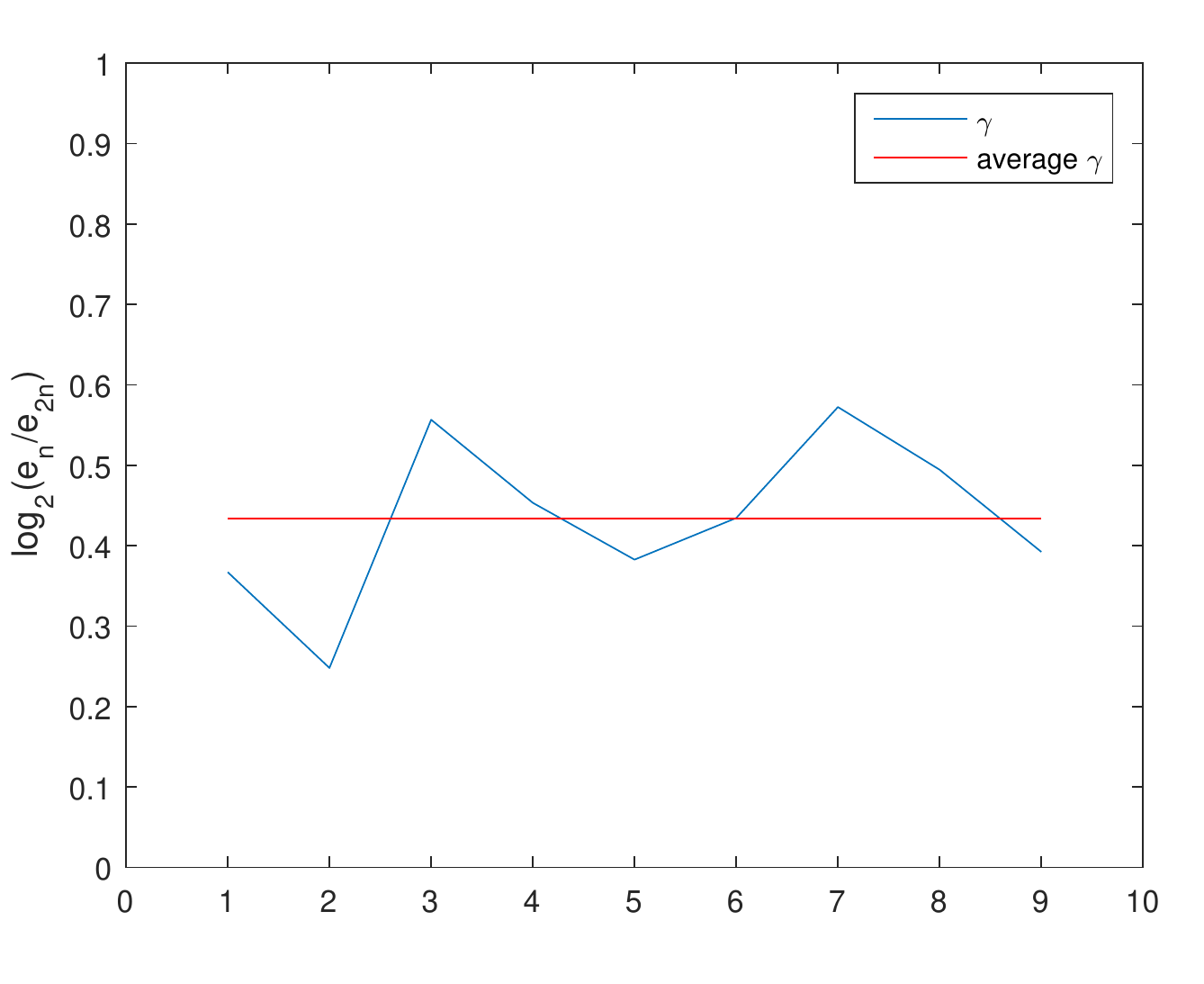}
\\
\end{tabular}

\end{center}
\caption{Simulation 4. On the left the graphic for the error $e_ n$ in function of the number of iterations. $ N_ {x} = N_ {y} = 50$, $T = 30$, $\left [y_ {min}, y_ {max} \right] = \left [-5.5 \right]$, 
so the area of each grid square is $\frac{1} {250}$; on the right the order of convergence $\gamma$, defined as $\log_{2}\left(\frac{e_{n}}{e_{2n}}\right)$. 
We performed 11 runs, with 100, 200, 400, 800, 1600, 3200, 6400, 12800, 25600, 51200, 102400 points.}
\end{figure}

\end{document}